\documentclass[aps, prx, 10pt, superscriptaddress, longbibliography, twocolumn]{revtex4-2}

\usepackage{array} 

\usepackage[table, dvipsnames]{xcolor} 
\usepackage[linesnumbered, ruled]{algorithm2e} 
\usepackage{amsmath} 
\usepackage{amssymb} 
\usepackage{amsthm} 
\usepackage{bm, bbm} 
\usepackage{booktabs} 
\usepackage{enumitem} 
\usepackage{geometry} 
\usepackage{hyperref} 
\usepackage{ifthen} 
\usepackage{mathrsfs} 
\usepackage{mathtools} 
\usepackage{multirow} 
\usepackage{mwe} 
\usepackage[sort&compress]{natbib} 
\usepackage{physics} 
\usepackage{subcaption}
\usepackage{thm-restate} 
\usepackage{tikz} 
\usepackage{xparse} 
\usepackage[capitalise]{cleveref} 
\usepackage{pifont}
\usepackage{cancel}

\hypersetup{colorlinks=true, linkcolor=blue, citecolor=blue, urlcolor=blue}

\definecolor{cA}{HTML}{0072BD}
\definecolor{cB}{HTML}{EDB120}
\definecolor{cC}{HTML}{77AC30}
\definecolor{cD}{HTML}{D95319}

\newcommand{\Cpp}{C\nolinebreak\hspace{-.05em}\raisebox{.4ex}{\tiny\bf +}\nolinebreak\hspace{-.10em}\raisebox{.4ex}{\tiny\bf +}}
\newcommand{\defeq}{\coloneqq}

\newtheorem{lemma}{Lemma}

\newtheorem{remark}{Remark}

\crefname{section}{Sec.}{Secs.}
\Crefname{section}{Section}{Sections}

\allowdisplaybreaks[4]

\begin{document}
\title{Bounded-depth spacetime lattice surgery for resource-efficient fault-tolerant quantum computation}

\author{Kou Hamada}
\email{kouhamada2002@gmail.com}
\affiliation{
Graduate School of Information Science and Technology,
The University of Tokyo,
7-3-1 Hongo, Bunkyo-ku, Tokyo 113-8656, Japan
}

\author{Hiroki Hamaguchi}
\email{hirokihamaguchi0331@gmail.com}
\affiliation{
Graduate School of Information Science and Technology,
The University of Tokyo,
7-3-1 Hongo, Bunkyo-ku, Tokyo 113-8656, Japan
}

\author{Yosuke Ueno}
\affiliation{
Graduate School of Information Science and Technology,
The University of Tokyo,
7-3-1 Hongo, Bunkyo-ku, Tokyo 113-8656, Japan
}
\affiliation{
Center for Quantum Computing, RIKEN,
2-1 Hirosawa, Wako, Saitama 351-0198, Japan
}

\author{Yasunari Suzuki}
\affiliation{
Center for Quantum Computing, RIKEN,
2-1 Hirosawa, Wako, Saitama 351-0198, Japan
}

\author{Teruo Tanimoto}
\affiliation{
Faculty of Information Science and Electrical Engineering,
Kyushu University,
744 Motooka, Nishi-ku, Fukuoka 819-0395, Japan
}

\author{Nobuyuki Yoshioka}
\email{ny.nobuyoshioka@gmail.com}
\affiliation{
International Center for Elementary Particle Physics,
The University of Tokyo,
7-3-1 Hongo, Bunkyo-ku, Tokyo 113-0033, Japan
}

\begin{abstract}
    Fault-tolerant quantum computing based on lattice surgery requires place-and-route compilation with low spacetime overhead. Routing, in particular, faces a basic tension between suppressing path conflicts through greater spatial allocation and exploiting the time direction to realize ancilla-efficient spacetime routing. Existing approaches do not fully resolve this trade-off while retaining compatibility with inner factory layouts and termination guarantees. Here we introduce \emph{double-slice routing}, a constant-depth spacetime-routing method that uses two consecutive time slices with a guarantee that its kink-parity correction terminates under both planar and stacked architectures. We numerically benchmark the resulting compiler on Hamiltonian-simulation workloads to show that double-slice routing reduces compilation cost by up to a factor of 2.4 over a single-slice baseline. Compared to projective routing, an existing method that allows an unbounded number of time slices per path, double-slice routing achieves smaller circuit volume with only a marginal execution-time penalty. Combined with a cultivation-compatible mapping optimization, the overall improvement reaches up to 7.5-fold over a naive single-slice compilation baseline. These results identify double-slice routing as a practically useful operating point in lattice-surgery compilation and show the substantial benefit in joint optimization of mapping and routing.
\end{abstract}

\maketitle

\section{Introduction}
Fault-tolerant quantum computing (FTQC) will require the compilation of quantum algorithms into
a restricted set of fault-tolerant instructions tailored to a given quantum error-correcting code and
hardware geometry.
Among the leading approaches, two-dimensional (2D) topological stabilizer codes such as the surface code~\cite{Kitaev2003Anyons, Fowler2012SurfaceCodes, Dennis2002TopologicalMemory}
and the color code~\cite{Bombin2007ColorCode} admit a particularly practical paradigm: \emph{lattice surgery}, where
joint Pauli measurements are realized by merging and splitting code patches
on a 2D grid~\cite{Horsman2012LS,Landahl2014ColorLS} which complements the restricted set of transversal operations by allowing gate teleportation.
Lattice surgery can also be viewed more broadly as a measurement-based form of code deformation (and, in many cases, gauge fixing),
providing a unifying language across patch-based FTQC schemes~\cite{Vuillot_2019}.

A central approach for compilation is based on the {\it place-and-route} paradigm: given a logical circuit, one must (i) assign logical data qubits and auxiliary resources (routing ancillae and magic state factories)
to locations on a qubit plane, and then (ii) schedule the instructions (including lattice-surgery routing) under spacetime constraints.
The quality of this compilation can significantly affect both the total execution time and the overall spacetime volume, which consequently impacts the achievable logical error rate at a fixed physical error budget.

The dominant bottleneck of FTQC depends on the available hardware scale and architectural choices.
In many large-scale resource estimates, the cost is heavily influenced by the supply of non-Clifford resources,
typically modeled through magic state preparation and distillation pipelines~\cite{litinskiGameSurfaceCodes2019,FowlerGidney2018LowOverhead}.
At the same time, recent progress in magic state preparation, including more efficient factory designs and direct fault-tolerant preparation protocols, suggests that there exist practically relevant regimes where routing and scheduling overheads become a comparable contributor to end-to-end cost~\cite{ChamberlandNoh2020DirectMagic}.
Indeed, efficient lattice-surgery scheduling has attracted growing attention: \citet{hamadaEfficientHighperformanceRouting2024} reduce scheduling to a 3D path-embedding problem and report substantial speedups over greedy baselines on FTQC benchmarks,
while related work studies multiqubit lattice-surgery scheduling via forest/Steiner-tree packing~\cite{silvaMultiqubitLatticeSurgery2024}.

\begin{figure*}
    \centering
    \includegraphics[width=\linewidth]{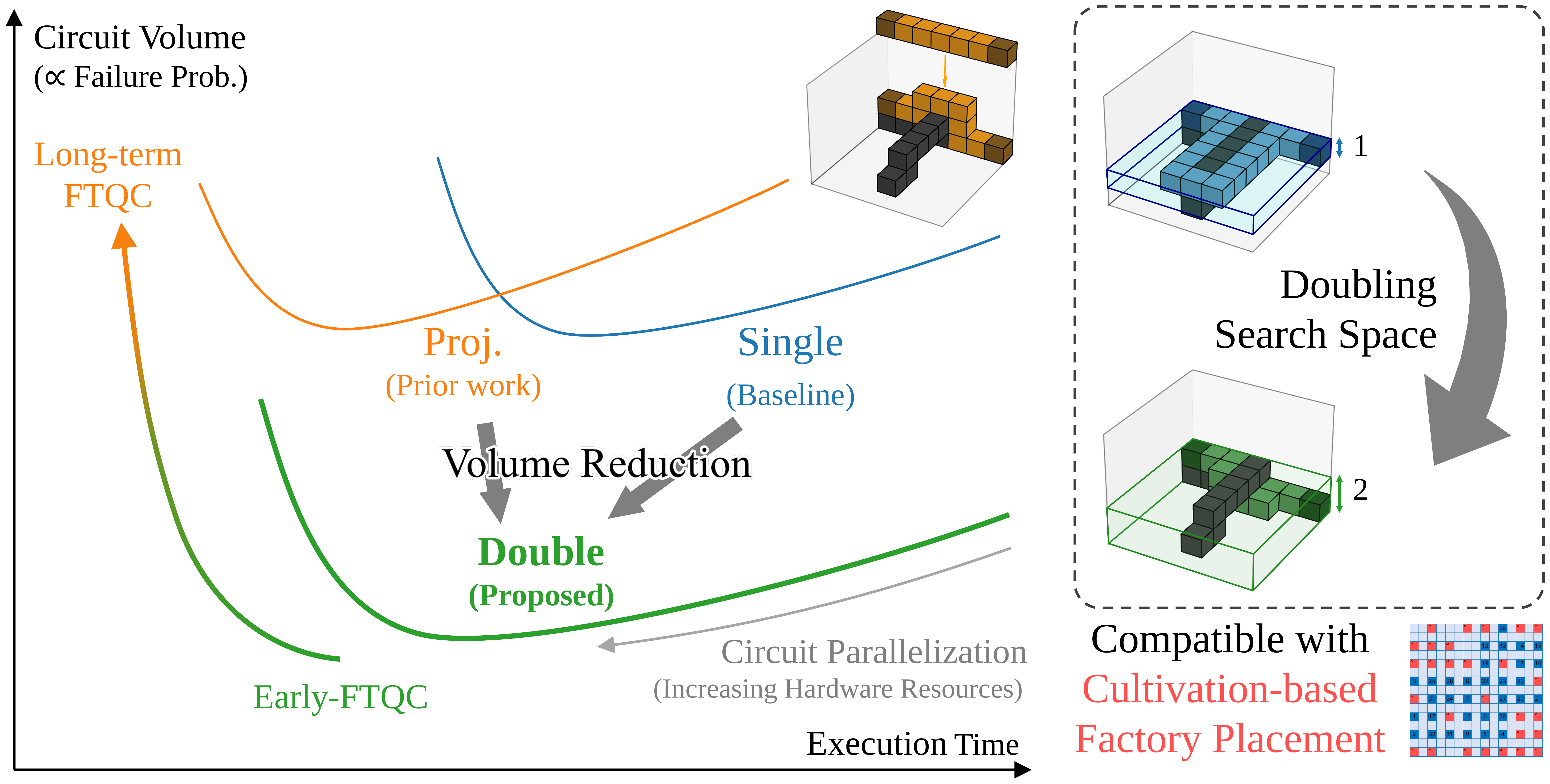}
    \caption{Schematic diagram of this paper. Double-slice routing is proposed as a spacetime routing method that \mbox{(i) reduces} the circuit volume, \mbox{(ii) provides} a theoretical guarantee on the termination of the compilation procedure, and \mbox{(iii) is} compatible with cultivation-based qubit mapping.
    }
    \label{fig:schematic}
\end{figure*}

In this work, we propose a compilation framework that flexibly exploits spacetime resources in both \emph{mapping} and \emph{routing} as described in \cref{fig:schematic}. 
Our main contributions can be summarized as follows.
\begin{itemize}
    \item First, we formulate a joint mapping problem that optimizes the placement of data patches, routing ancillae, and magic state factory (MSF) patches in a unified manner.
    Motivated by the recent advancements in magic state cultivation~\cite{gidneyMagicStateCultivation2024}, this formulation enables cultivation-compatible inner factory layouts and extends prior placement-oriented approaches~\cite{laoMappingLatticeSurgerybased2018}.
    \item Second, we develop a spacetime-efficient routing framework for both 2D and 2.5D (or multiplanar) architectures. This routing contribution has two components: generalized spacetime routing for 2.5D layouts and introducing {\it double-slice routing}, a constant-depth method that mitigates path collisions. We show that the double-slice routing offers both quantitative and qualitative benefits: reduction in spacetime volume and compatibility between cultivation-aware inner layout and the termination guarantee of path finding. To our knowledge, our work is the first to reconcile the executable guarantee for spacetime routing and cultivation-aware geometry. 
\end{itemize}

We provide numerical evaluation on our approach for Hamiltonian-simulation workloads that are widely regarded as promising candidates for demonstrating early advantages of FTQC, and observe consistent reductions in compilation cost.
Across a range of architectures and benchmarks, our mapping optimization achieves up to 3.2$\times$ improvements over baselines, and double-slice routing yields up to 2.4$\times$ reductions in spacetime volume compared with single-slice strategies.
The overall improvement reaches up to 7.5-fold over the naive baseline,
highlighting a practical space--time trade-off in lattice-surgery compilation.

The rest of the paper is organized as follows.
\begin{itemize}
    \item \Cref{sec:setup} introduces the fundamental concepts of surface code computation and decomposes the lattice-surgery scheduling problem into two primary subproblems: the mapping and routing problems.

    \item \Cref{sec:mapping} addresses the mapping problem, which involves assigning circuit operands and MSFs to specific patches on the qubit plane. We propose an optimization algorithm based on Simulated Annealing (SA) that also accounts for the recent magic state cultivation paradigm. This technique is applicable to \textit{inner factory layouts}, in which MSFs can be embedded into the bulk region of the qubit plane.

    \item \Cref{sec:routing} focuses on the routing problem, a problem of determining efficient lattice-surgery paths. We extend spacetime routing to 2.5D architectures and introduce \textit{double-slice routing}, which combines the strengths of baseline methods and existing \textit{projective routing} while maintaining compatibility with inner factory layouts.

    \item \Cref{sec:numerical_experiments} presents numerical experiment results to evaluate the proposed algorithms. We demonstrate their effectiveness in terms of total execution time and circuit volume using the Hamiltonian-simulation circuits based on quantum singular value transformation and Trotterization.

    \item \Cref{sec:discussion} concludes the paper and suggests future directions.
\end{itemize}

The code used for the numerical experiments, mainly written in Python 3 and \Cpp{17}, is fully available on GitHub~\cite{github}.

\section{Problem Setup} \label{sec:setup}

In this section, we formalize the lattice-surgery scheduling problem.
We provide an overview of the problem, followed by a detailed discussion on performance metrics, hardware constraints, and the benchmarked quantum circuits.

\subsection{Lattice-Surgery Scheduling Problem}\label{subsec:lssp}

In the following, we adopt a patch-level abstraction of surface-code-based FTQC on a fixed grid architecture motivated by superconducting devices: the hardware provides a finite set of logical surface-code patches of a fixed code distance $d$ arranged on a grid geometry.
Time is discretized in units of $d$ code cycles, or code beats, and non-local logical operations are realized by lattice surgery, i.e., joint Pauli measurements performed between patch boundaries using intermediate ancilla resources, which are referred to as the bus qubits.

Throughout the entire execution of the quantum circuit, we assume that each logical patch is assigned a single role: data patches that store logical information (including the ancilla for block encoding), bus patches that provide routing resources to mediate joint measurements, and MSF patches that supply magic states at a fixed preparation rate.
Under this static role assignment, a lattice-surgery operation temporarily occupies a set of patches, namely its operand patches together with a connected chain of bus patches.
Thus, concurrent operations must respect resource exclusivity on the shared grid (see \cref{subsec:hardware} for architectural details).

Starting from a quantum circuit given in the Clifford+T gate set, we assume a standard compilation step that reduces non-local logical operations to a small set of lattice-surgery-native QEC micro-operations.
In this work, we define the primitive multiqubit QEC micro-operations as
(i) joint $XX$ measurements, (ii) joint $ZZ$ measurements, and (iii) logical patch movement.
Higher-level instructions are realized as compositions of these primitive QEC micro-operations; for instance, the $T$ gate is implemented through gate teleportation using a magic state. While we account for the cost of magic state preparation, we assume that each magic state is prepared immediately before the corresponding lattice-surgery operation is performed, thereby incorporating the preparation overhead into the routing path cost. 

In contrast, we omit local Clifford gates (e.g., $H$ and $S$ gates) in this scheduling formulation, because their spatial and temporal overheads are negligible compared to those of lattice-surgery operations, particularly in the context of large-scale circuits. Consequently, every instruction involves at most two non-bus operand patches in total (data and/or factory), while the number of intermediate bus patches used to connect them is not bounded {\it a priori} by the instruction itself. See \cref{app:conversion} for a detailed discussion on this compilation and formulation.

The lattice-surgery scheduling problem can be informally stated as follows.
For each instruction in the given program, decide its execution timing and allocation of bus patches that form a valid connection between its operands. The solution must satisfy constraints such as geometric connectivity, boundary-type compatibility, and the exclusivity conditions imposed by the finite patch grid.

We solve the scheduling problem by dividing it into the following two subproblems.
\begin{enumerate}
    \item {\it Mapping problem (\cref{sec:mapping})}: Assign the role of every logical patch, under assumptions that (i) the total number of logical patches is fixed and (ii) each patch is dedicated to a specific role (data, bus, or magic) during the entire quantum computation.
    \item {\it Routing problem (\cref{sec:routing})}: Determine the most efficient lattice surgery paths, under hardware constraints explained in \cref{subsec:hardware}.
\end{enumerate}

\begin{figure*}[t]
    \centering
    \includegraphics[width=1.9\columnwidth]{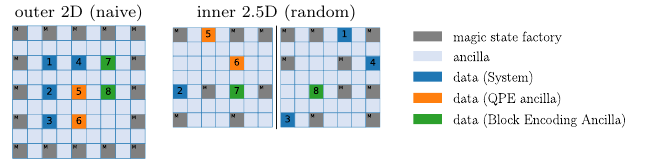}
    \caption{Examples of layouts. In the 2D layout, the data chip consists of a single qubit plane structured as an array of surface-code patches, whereas in the 2.5D layout, an additional qubit plane is stacked on top. In the outer factory layout, MSF patches are positioned along the perimeter, whereas in the inner factory layout, they are placed within the interior.         The numbers labeled on the data qubits indicate their indices.}
    \label{fig:layout}
\end{figure*}

\subsection{Performance Metrics}
\label{sec:performance_metric}

For the evaluation of the performance of compilation methods, we consider two metrics: {\it total execution time} and  {\it total circuit volume}. While the former directly reflects the overall time required to complete the quantum computation, the latter is also essential since it serves as a proxy for the total logical error rate; furthermore, since decoding cost scales super-linearly with volume, this metric directly dictates the classical computational burden. Thus, our primary objective is to minimize both metrics through mapping and routing optimization.

To facilitate the determination of lattice-surgery paths, we employ the spacetime routing technique~\cite{hamadaEfficientHighperformanceRouting2024} (see \cref{subsec:overview_3drouting}), which associates a lattice-surgery instruction with a spacetime path.
Consequently, we hereafter treat a lattice-surgery path as a spacetime path.
Let $X$ denote the set of patches.
We define a {\it voxel} as a spacetime unit $(x, t)$ activated by a path, where $x \in X$ is the spatial coordinate and $t \in \mathbb{N}$ is the activation time measured in code beats (Recall that $1$ code beat equals $d$ code cycles).

We first define the total execution time.
Let $K$ be the number of instructions, and for the $k$-th instruction ($1 \leq k \leq K$), we define the associated path $P_k$ as
\begin{equation*}    
    P_k = \{(x_{k, i}, t_{k, i})\}_i \subseteq X \times \mathbb{N},
\end{equation*}
where $(x_{k, i}, t_{k, i})$ denotes the $i$-th activated voxel of the $k$-th instruction.
Given the collection of all $K$ paths $P \defeq (P_1, \dots, P_K),$ the total execution time $T_P$ is defined  as
\begin{equation*}
    T_P \defeq \max_{1 \leq k \leq K} \; \max_{(x,t)\in P_k} t,
\end{equation*}
which means the last activation time over all paths.

We next define the total circuit volume.
Let $T_{\mathrm{LS}}(x)$ be the set of times when the patch $x$ is contained in lattice surgery:
\begin{equation*}
    T_{\mathrm{LS}}(x) \defeq \bigcup_{k=1}^{K} \{ t \mid (x, t) \in P_k\}.
\end{equation*}
Then, the circuit volume $V(x)$ for each patch is
\begin{equation*}
    V(x) \defeq
    \begin{cases}
        T_P                                      & \text{if $x$ is a data patch,} \\
        \abs{T_{\mathrm{LS}}(x)}\cdot (\tau + 1) & \text{if $x$ is an MSF patch,} \\
        \abs{T_{\mathrm{LS}}(x)}                 & \text{if $x$ is a bus patch,}
    \end{cases}
\end{equation*}
where $\tau$ is defined to be the fixed time interval required to prepare a single magic state.
For a data patch $x$, its circuit volume is the total execution time $T_P$, since error correction is required throughout the entire computation to maintain its data.
For an MSF patch $x$, the circuit volume accounts for the total cost of magic state preparation and consumption.
For a bus patch, its circuit volume is simply the number of time steps it is active.
Then, the total circuit volume $V_P$ is defined as the sum over the volume of individual patches:
\begin{equation*}
    V_P \defeq \sum_{x \in X} V(x).
\end{equation*}

\subsection{Hardware Constraints} \label{subsec:hardware}

When we consider the lattice-surgery scheduling problem, there is usually an implicit assumption on the hardware constraint that physical qubits cannot be moved around.
Specifically, we describe the geometrical arrangement of logical patches, boundary types of surface-code patches, and constraints on the possible roles of patches.

{\it Geometry of logical patches. }
Reflecting the experimental feasibility in a superconducting platform, it is commonly assumed that surface code patches are arranged on a 2D square lattice. Here, we refer to such an arrangement as the {\it 2D layout}.
An alternative emerging design is the {\it 2.5D layout}~\cite{liao2026breakingscalabilitybarriervertical}, in which a couple of 2D layers are stacked to improve the routing efficiency via benefiting from transversal CNOT operation or inter-layer lattice surgery~\cite{uenoHighPerformanceScalableFaultTolerant2024,viszlaiInterleavedLogicalQubits2025,pattisonHierarchicalMemoriesSimulating2025,rametteFaulttolerantConnectionErrorcorrected2024,haArchitecturesLatticeSurgerybased2025}. Although our discussion of 2.5D layouts primarily focuses on the biplanar architecture due to its experimental viability, our proposed approach is independent of this specific structure. Consequently, our theoretical framework can be extended to multiplanar configurations in most cases, depending on the architectural choices.
For both geometries, we adopt the conventional density of $1/4$ for data and MSF patches, as shown in the examples of \cref{fig:layout}.
The patches are arranged on a square grid, with the grid size chosen separately for each layer count $l$ as the minimum size required to accommodate all logical qubits, thereby making the comparison as fair as possible.
See \cref{app:other_architectures} for more details on our architectural assumptions.

{\it Boundary types of surface-code patches. }
Each edge of a surface-code patch supports only one type of syndrome measurement, either $X$ or $Z$.
We refer to these edges as the \emph{$X$-boundary} or \emph{$Z$-boundary}, based on their measurement type. A surface code patch allows for two possible boundary orientations upon initialization: (i) the top and bottom edges are $X$-boundaries, and the left and right edges are $Z$-boundaries; or (ii) the top and bottom edges are $Z$-boundaries, and the left and right edges are $X$-boundaries. To realize a joint $XX$ (resp.\ $ZZ$) measurement, both operand patches should be connected to the bus patches of the path along an $X$-boundary (resp.\ $Z$-boundary).

For simplicity, we assume that the boundary orientation of each patch remains fixed once it is initialized. We also assume that all the data and MSF patches share the same boundary orientation. In contrast, the bus patches can be initialized with either boundary orientation.

{\it Role of patches in the bulk of the grid.}
Each logical patch is assigned a fixed role, either data, bus, or MSF, throughout the execution of the quantum circuit. While data and bus patches may in principle be placed anywhere on the qubit plane, the placement of MSFs can require additional constraints. We refer to the rule governing the placement of MSFs as the factory layout.

We first consider an \textit{outer factory layout}, in which the MSF patches are restricted to the perimeter of the qubit plane. This restriction is consistent with several existing works~\cite{chamberlandUniversalQuantumComputing2022,beverlandSurfaceCodeCompilation2022}. The motivation is that magic state distillation typically incurs a substantial spatial overhead~\cite{litinskiGameSurfaceCodes2019,litinski2019magic}. Consequently, it is natural to place large distillation factories outside the main computational region, while supplying distilled magic states to the data region from the boundary.

In contrast, we also consider an \textit{inner factory layout}, in which the MSF patches may be allocated within the lattice similarly to data patches (\cref{fig:layout}).
This factory layout is motivated by the recently proposed technique of magic state cultivation~\cite{gidneyMagicStateCultivation2024} requiring only a single patch to prepare the magic state. Such compact implementations naturally enable the placement of MSFs inside the computational region itself.

For both layouts, we assume for simplicity that each MSF generates one magic state at a fixed rate. The required time interval is defined as the magic state preparation time, denoted by $\tau$ in code beats. We set the magic state preparation time as $\tau = 2$, with the scope of the discard rate becoming lower as the physical error rate improves. See \cref{subsubsec:tau_comparison} for the discussion on the choice of $\tau$.

Regarding outer factory layouts, $\tau = 0$ is another reasonable choice as employed in Ref.~\cite{beverlandSurfaceCodeCompilation2022}. This setting is particularly useful for highlighting the performance of the routing algorithm under ideal conditions, where external resources are assumed to supply magic states sufficiently at the perimeter of the qubit plane. We adopt this configuration in specific cases, for example, when comparing our scheduling algorithms with Ref.~\cite{beverlandSurfaceCodeCompilation2022} (see \cref{sec:numerical_experiments}).

\subsection{Benchmarked Quantum Circuits}
\label{subsec:benchmark_target}
As a benchmark target, we exclusively consider quantum algorithms that aim to simulate quantum many-body systems: the quantum singular value transformation~\cite{low2017optimal, low2019hamiltonian, gilyen2019quantum, martyn_2021} and Trotterization~\cite{suzuki1985general,childs_theory_2021}. For the former, we specifically focus on the cost of implementing the SELECT circuit for block encoding of a Hamiltonian, and for the latter, we consider second-order Trotterization.
Target Hamiltonians are chosen from versatile fields of physics, such as the Heisenberg model, the Fermi--Hubbard model, lattice gauge theory, and also the random local Hamiltonian. Also, see \cref{app:target} for details on the quantum algorithms and target Hamiltonians.
While in the main text we mainly discuss the result for the $J_1$-$J_2$ Heisenberg model on a square lattice, we provide details in the Appendix with data available via GitHub~\cite{github}.

Since Ref.~\cite{yoshiokaHuntingQuantumclassicalCrossover2024} demonstrated that parallelization is effective for SELECT circuits, we focus primarily on their parallelized versions.
Specifically, we refer to a SELECT circuit parallelized into $2^i$ threads as SELECT-$i$. We mainly employ $i = 5$ for its reasonable balance between execution time and resource requirement, while also utilizing $i = 6$ as a more resource-intensive case to highlight performance differences.

\section{Mapping Optimization}\label{sec:mapping}

As explained in \cref{sec:setup}, we decompose the lattice-surgery scheduling problem into two phased subproblems: mapping and routing problems. We begin the explanation with the mapping problem.

Mapping refers to determining the assignment between indexed data qubits and logical patches, as well as selecting which logical patches are used as MSF patches.
In this section, we introduce the optimization of mapping with simulated annealing, which drastically reduces the total execution time.
The mapping takes into account the recent advancement in the magic state cultivation \cite{hiranoLocalityAwarePauliBasedComputation2025}, i.e., we consider the inner factory layout where we can freely locate the MSFs in the qubit plane.

Considering the inner factory layout instead of the outer factory layout means that we remove the perimeter constraint of MSFs. It increases the search space and makes the mapping problem more flexible, but also more complex.
We show that the proposed simulated-annealing-based optimization can efficiently handle this generalized mapping problem and substantially reduce the total execution time.
In particular, for the 2.5D layout, the optimized inner factory layout nearly saturates the achievable reduction in total execution time.

\subsection{Problem Setting}\label{subsec:sa_setting}

To begin with, we mathematically formulate the higher-level instructions introduced in \cref{subsec:lssp}.
As detailed in \cref{app:conversion}, we convert quantum circuits into instruction sequences consisting solely of $\texttt{CX}$, $\texttt{MAGIC\_MOVE}$, and $\texttt{MAGIC\_MZZ}$, representing CNOT operations, magic-state movement operations, and Pauli-$ZZ$ measurements assisted by magic states, respectively.
These instructions only involve two patches and form pairwise connections realized as lattice-surgery paths.

Now, to evaluate the frequency of interactions among qubits in the quantum circuit, we define a heatmap $A$.
Although multiple MSF patches may exist, the instructions that consume a magic state do not specify which patch supplies the magic state.
We therefore model all MSFs as a single virtual qubit for the mapping optimization.
Let $n$ be the number of logical data qubits, where $Q = \{1, 2, \dots, n\}$ represents the set of these qubits, and let $n+1$ denote the virtual MSF.
Accordingly, the $k$-th instruction can be regarded as acting on a pair of qubits, which may include the virtual MSF qubit.
The two qubits involved are denoted by $p_{k,1}$ and $p_{k,2}$, with $1 \leq p_{k,1}, p_{k,2} \leq n+1$.
We define the $(i,j)$-th component of $A$ as the number of occurrences of the pair $\{i,j\}$:
\begin{equation*}
    A_{i,j}\defeq \sum_{k=1}^{K}\mathbbm{1}\left[\{p_{k,1}, p_{k,2}\}=\{i, j\}\right] \quad (1 \leq i,j \leq n+1).
\end{equation*}
Here, $\mathbbm{1}$ represents an indicator function.
The matrix $A$ indicates how frequently each data qubit is simultaneously involved with another qubit or with an MSF, thereby capturing structural properties of the circuit.

The key property of the heatmap $A$ is its sparsity.
For randomly generated circuits, this heatmap typically produces a dense, unstructured matrix.
In contrast, practical circuits such as the QSVT algorithm often exhibit sparsity.
The SELECT circuit for block encoding of a Hamiltonian consists of controlled Pauli gates, where the control register repeatedly interacts with the same patch in practice.
See also \cref{fig:heatmap}.
Also, for the Trotter circuit, a leading candidate for both phase estimation and dynamics simulation, sparsity is also commonly observed in models with geometrically local interactions.
Consequently, by pre-placing qubits, frequently interacting qubits can be located close to one another, whereas infrequently interacting qubits may be placed farther apart, thereby improving lattice surgery optimization.

\begin{figure}[t]
    \centering
    \includegraphics[width=0.9\columnwidth]{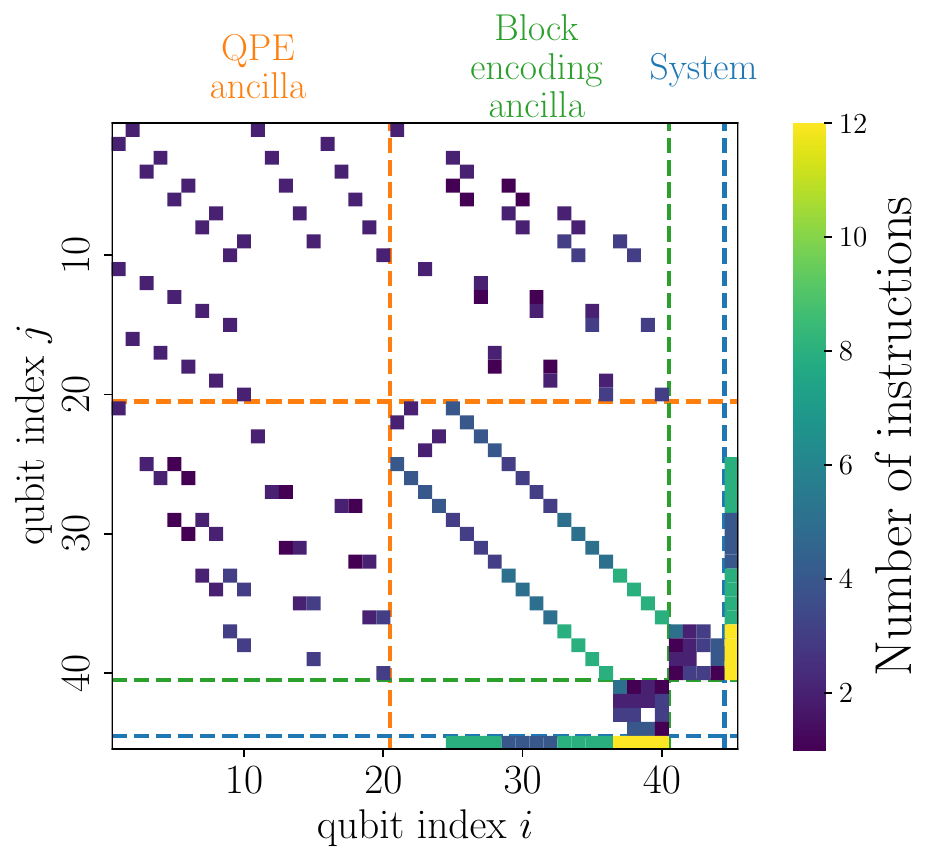}
    \caption{
        Heatmap of $A$ (\cref{subsec:sa_setting}) showing the number of instructions.
        In the SELECT-2 circuit simulating a 2D Heisenberg model (size = 2), the first 20 qubits correspond to QPE ancilla, the next 20 to block-encoding ancilla, and the last 4 to system qubits.
        The index $i = n+1$ corresponds to MSFs.
        Blank entries denote zeros.
        The matrix $A$ is highly sparse, which affects the optimization of the mapping.
    }
    \label{fig:heatmap}
\end{figure}

We now define an optimization variable for the mapping problem.
The subset $\bar{X} \subseteq X$ denotes the patches available for the placement of logical data qubits and MSFs, whose size is approximately 1/4 that of $X$ under our adopted qubit density.
The optimization variables are $\{x_i\}_{i=1}^n$ and a subset $X_{\mathrm{MSF}} \subseteq \bar{X}$, where $x_i \in \bar{X}$ denotes the patch to which data qubit $i$ is assigned, and $X_{\mathrm{MSF}}$ represents the patches occupied by MSFs.
Since the number of MSFs can influence the compilation results, we fix the number of MSFs $\abs{X_{\mathrm{MSF}}}$ to $n_{\mathrm{MSF}} \in \mathbb{N}$ in advance to ensure a fair comparison between 2D and 2.5D layouts.
Since each patch can host at most one object, all assignments must be distinct.
In the outer factory layout, we have to impose the additional constraint $X_{\mathrm{MSF}} \subseteq \partial \bar{X}$, where $\partial \bar{X}\subseteq X$ denotes the set of perimeter patches of the chip.
In contrast, the inner factory layout is obtained by removing this perimeter constraint.
More detailed specifications are provided in \cref{app:floorplan}.

The objective function for the mapping problem is formulated using the heatmap $A$.
The ideal objective function is the total code beats, but it is computationally expensive for large instances.
Instead, we optimize a cost function that captures the dominant geometric effects while remaining computationally tractable.
The objective function consists of two types of geometric costs.
We define $d_{\mathrm{data}}(x_i,x_j)$ as the Manhattan distance in the qubit plane between two data qubit patches $x_i$ and $x_j$.
We also define $d_{\mathrm{MSF}}(x_i,X_{\mathrm{MSF}})$ as the minimum Manhattan distance from a data-qubit patch $x_i$ to any MSF patch in $X_{\mathrm{MSF}}$.
Since the matrix $A$ encodes the frequency of interactions between qubits, the total path length required to compile each instruction into lattice-surgery operations can be approximately estimated by the following costs:
\begin{equation}
    \begin{dcases}
        A_{i,j}\, d_{\mathrm{data}}(x_i, x_j),
         & (\text{data--data cost}) \\
        c_{\mathrm{MSF}}\, A_{i,n+1}\, d_{\mathrm{MSF}}(x_i, X_{\mathrm{MSF}}).
         & (\text{data--MSF cost})
    \end{dcases}
\end{equation}
Here, $c_{\mathrm{MSF}} > 0$ is an MSF coefficient that tunes the relative importance of data--MSF instructions.

Using these definitions, we formulate the problem as:
\begin{align}
\underset{
\substack{\{x_i\}_{i=1}^n \subseteq \bar{X},\\X_{\mathrm{MSF}} \subseteq \bar{X}
}}{\mathrm{minimize}} \quad
& \sum_{i=1}^{n} \sum_{j=i+1}^{n} A_{i,j} d_\mathrm{data}(x_i, x_j) \label{prob:mapping_optimization_problem} \\
&{} + c_{\mathrm{MSF}} \sum_{i=1}^{n} A_{i,n+1} d_\mathrm{MSF}(x_i, X_{\mathrm{MSF}}) \notag \\
\mathrm{subject\ to}\quad&\abs{\{x_i\}_{i=1}^n} = n, \quad \abs{X_{\mathrm{MSF}}} = n_{\mathrm{MSF}}, \notag \\
&\{x_i\}_{i=1}^n \cap X_{\mathrm{MSF}} = \emptyset, \notag \\
&X_{\mathrm{MSF}} \subseteq \partial \bar{X} \quad \text{for the outer factory layout.} \notag
\end{align}
The first term of this formulation is essentially equivalent to the quadratic
assignment problem proposed in~\cite{laoMappingLatticeSurgerybased2018}.
We also note that the formulation naturally extends to 2.5D architectures, including multiplanar configurations.
One may also perform optimization based on the area of bounding boxes that involve the operands~\cite{molaviDependencyAwareCompilationSurface2025}.

In addition to our proposed solution, several simple baselines can be considered for Problem~\eqref{prob:mapping_optimization_problem}.
The first is a \textit{naive} mapping, in which $x_i$ and $X_{\mathrm{MSF}}$ are assigned sequentially according to the qubit indices.
Another is a \textit{random} mapping obtained by randomly shuffling this assignment.
These baselines correspond to straightforward implementation choices and are therefore useful for comparison.
Examples of these mappings are shown in \cref{tab:mapping_images}.

\begin{table}[t]
    \centering
    \caption{Comparison of mapping allocation methods for inner and outer factory layouts.}
    \label{tab:mapping_images}
    \begin{tabular}{cccc}
        \toprule
                                                            & \hfil naive \hfil & \hfil random \hfil & \hfil SA \hfil \\
        \midrule
        \rotatebox{90}{\makebox[0.3\columnwidth][c]{outer}} & \includegraphics[width=0.3\columnwidth]{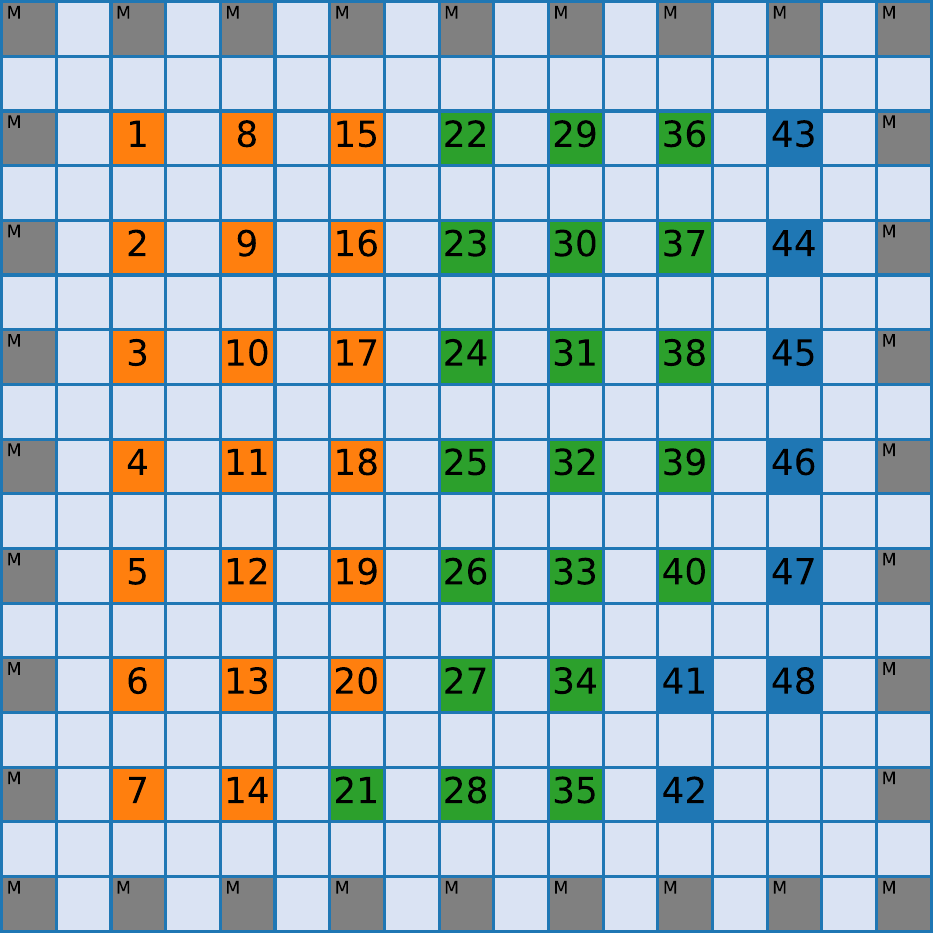} & \includegraphics[width=0.3\columnwidth]{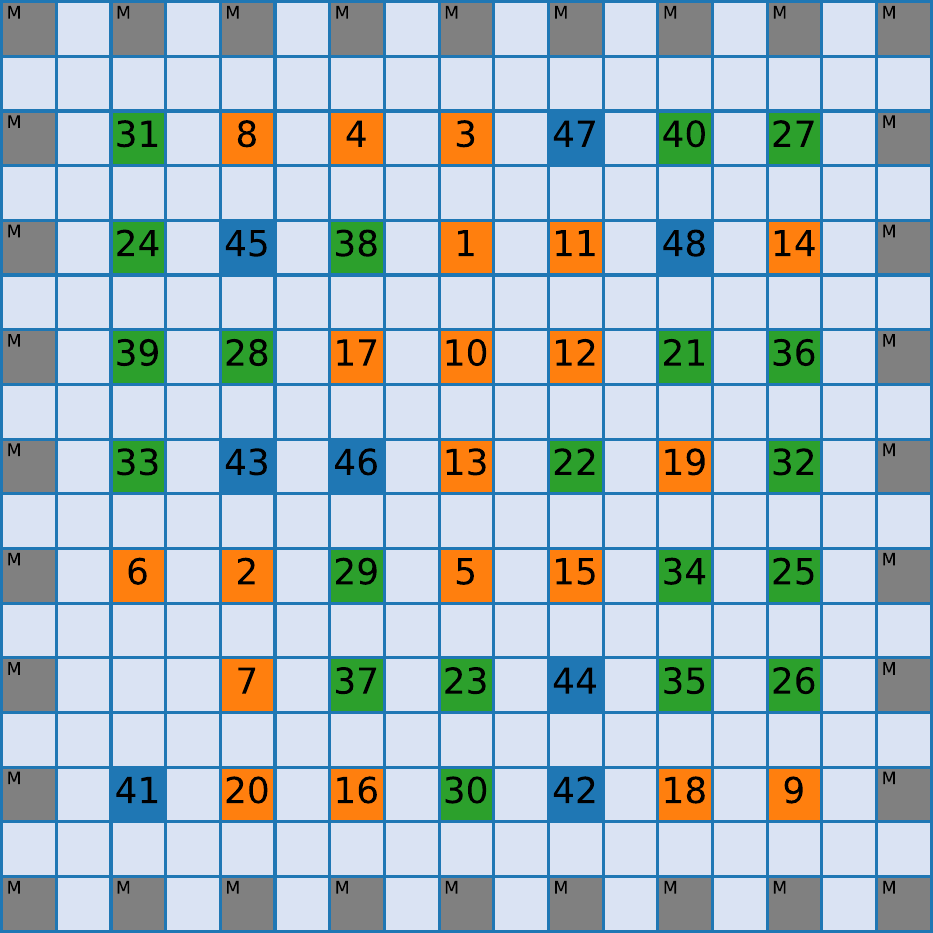} & \includegraphics[width=0.3\columnwidth]{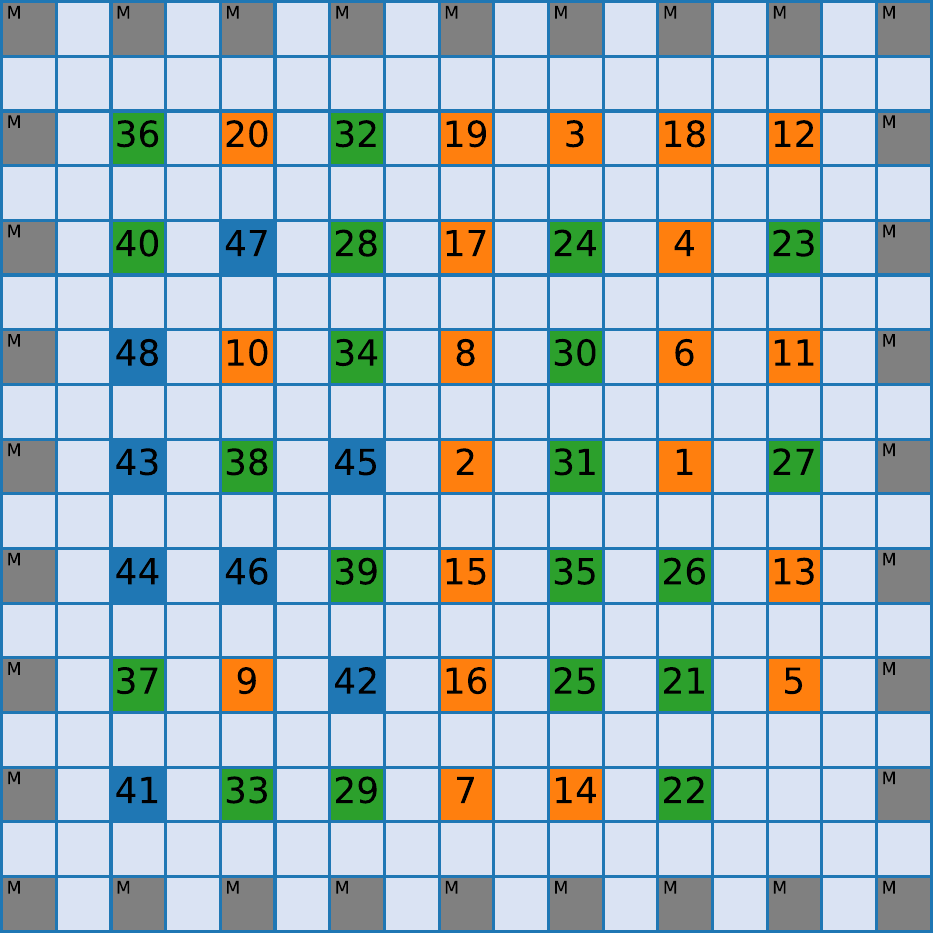} \\
        \rotatebox{90}{\makebox[0.3\columnwidth][c]{inner}} & \includegraphics[width=0.3\columnwidth]{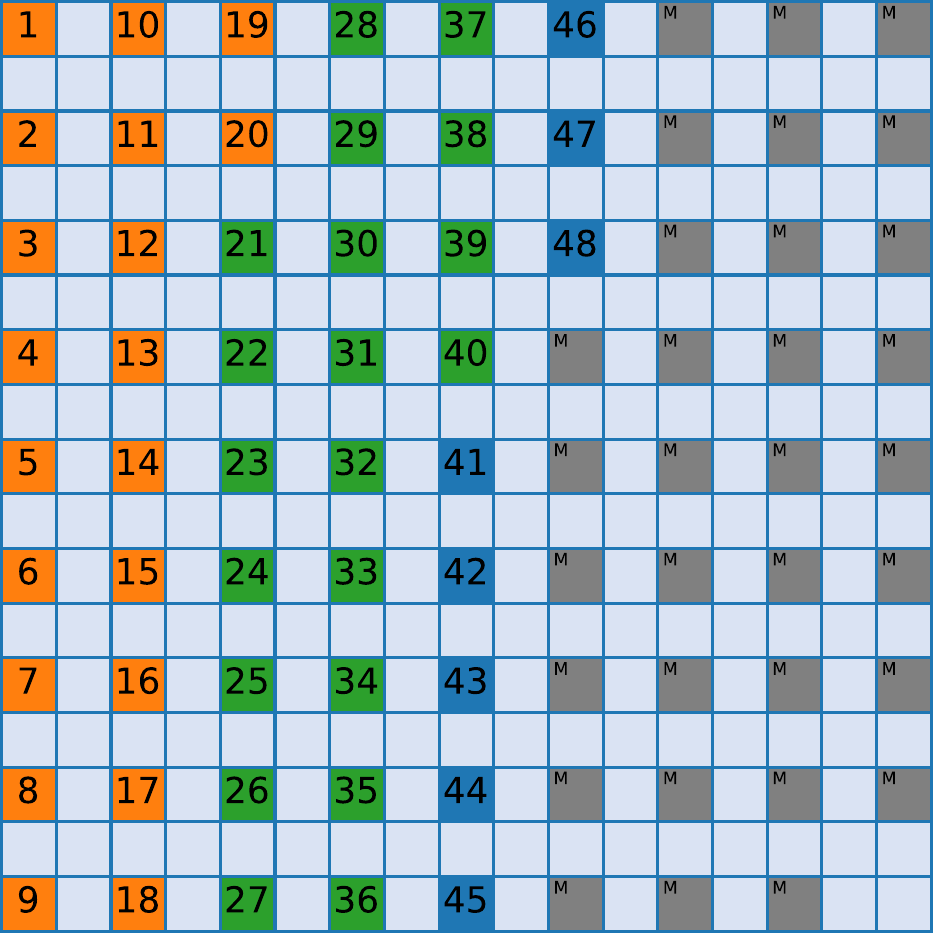} & \includegraphics[width=0.3\columnwidth]{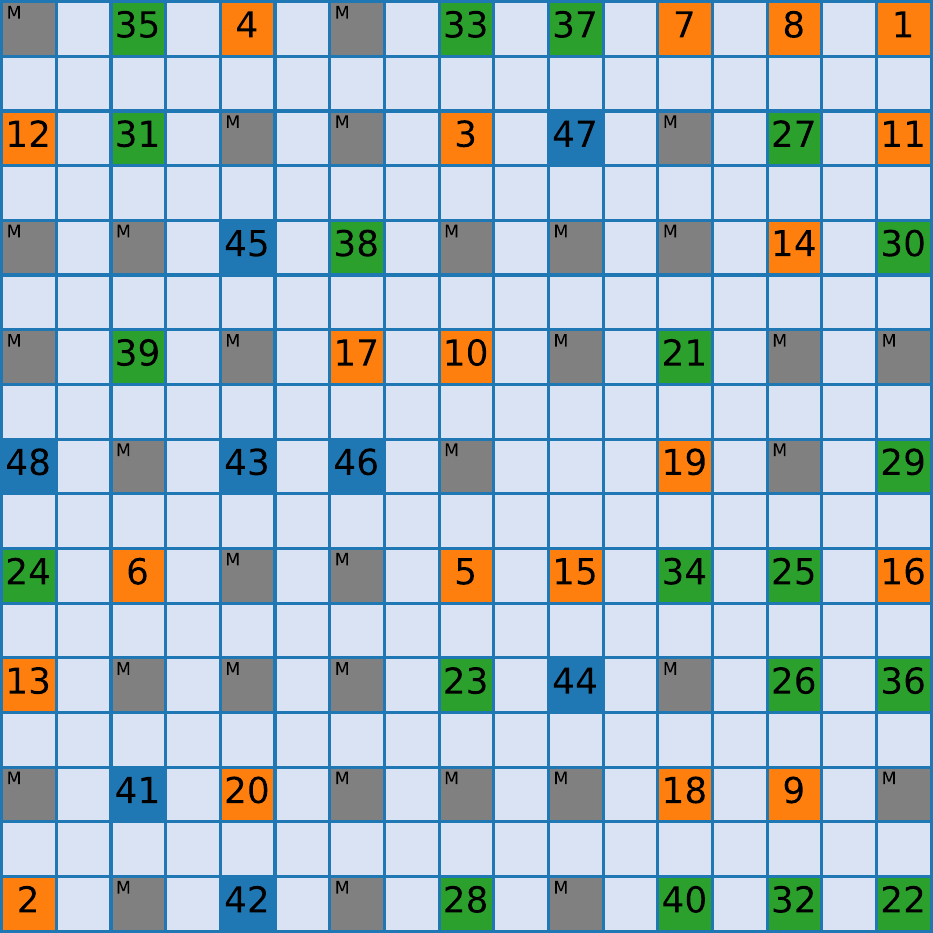} & \includegraphics[width=0.3\columnwidth]{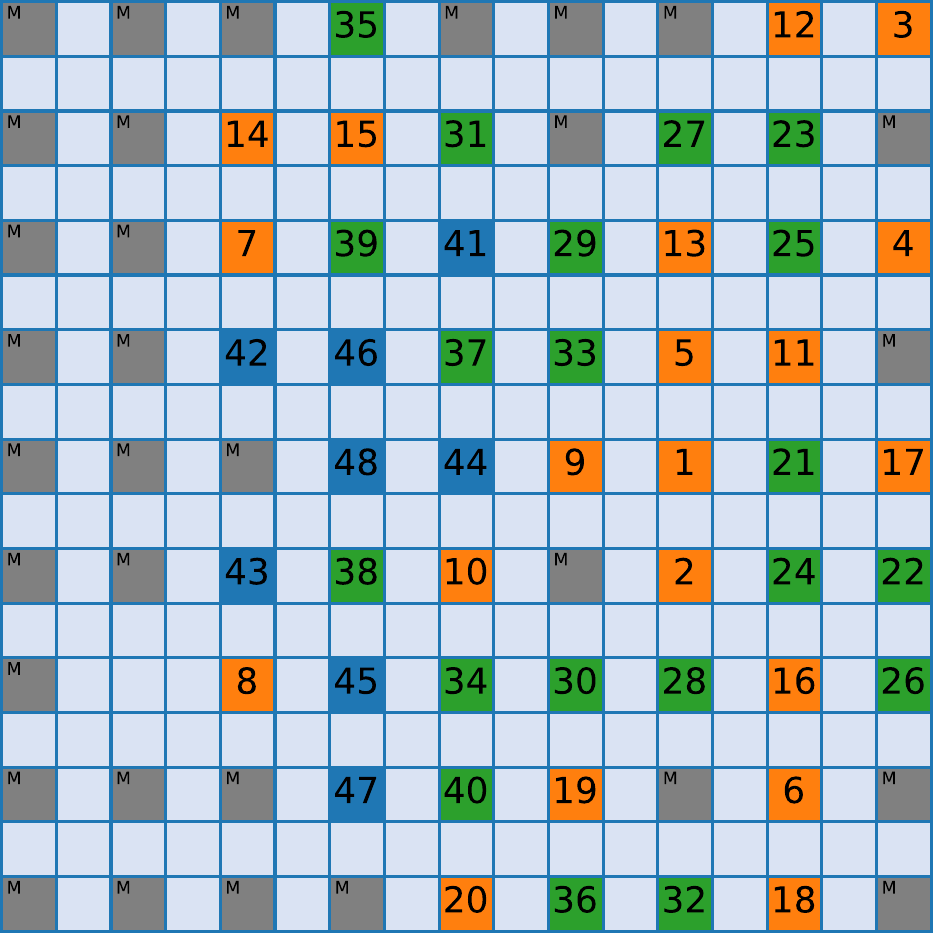} \\
        \bottomrule
    \end{tabular}
\end{table}

\subsection{Solution}

We optimize Problem~\eqref{prob:mapping_optimization_problem} by Simulated Annealing (SA) and obtain a mapping. SA is a classical meta-heuristic optimization method. In this work, we employ SA to obtain an approximate solution.
In particular, the optimization is carried out on a classical computer with computational complexity independent of the number of instructions $K$, which makes the method scalable and applicable to large-scale problems. The details of our implementation of SA are provided in \cref{app:mapping_how_to_solve}.

In our experiments, we tested four parameter settings $c_{\mathrm{MSF}}\in \{1, 0.1, 0.01, 0.001\}$ and selected the configuration that minimized the total execution time. While further tuning of this parameter could potentially lead to better solutions, this comes at the cost of additional computational effort on a classical computer, resulting in a practical trade-off.

\subsection{Experimental Results}

\begin{figure}[t]
    \centering
    \includegraphics[width=\columnwidth]{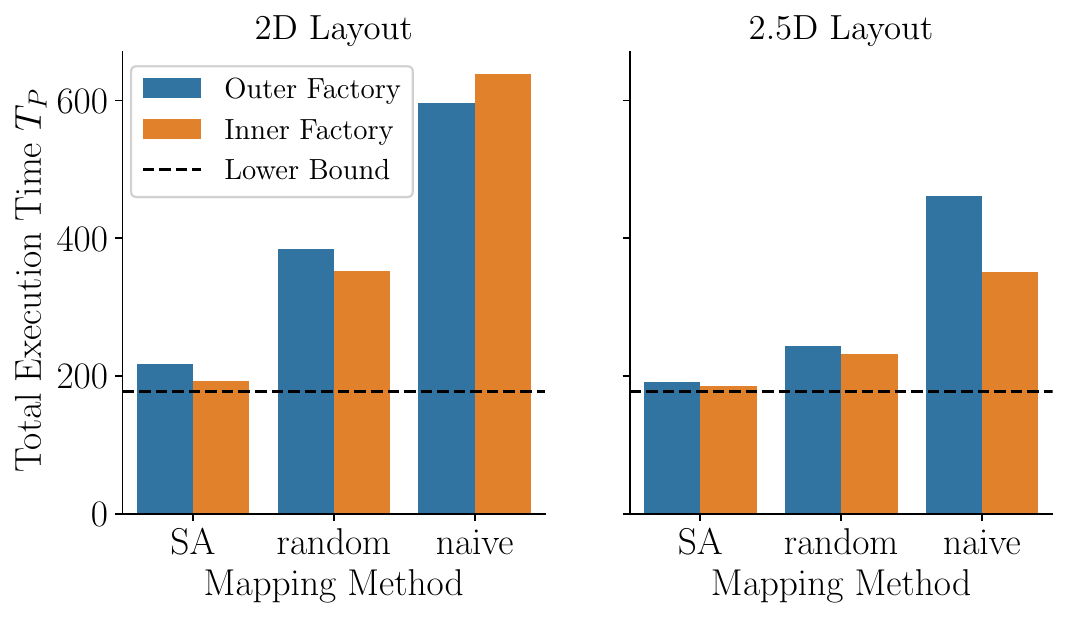}
    \caption{Total execution time $T_P$ for the SELECT-6 circuit simulating the 2D Heisenberg model (size = 10) in both 2D and 2.5D layouts, where $\tau$ is set to $2$. We used double-slice routing for scheduling. The performance follows the order: SA, random, and naive, from best to worst in almost all instances. The SA method achieves performance very close to the attainable lower bound.}
    \label{fig:total_time}
\end{figure}

We present partial experimental results on the mapping.
The key finding is that, when combined with the routing method introduced in the next section (\cref{sec:routing}), SA almost always reduces the total execution time.
As shown in \cref{fig:total_time}, combining SA-based mapping with the scheduling method almost attains the lower bound of total execution time.
Similar results hold for nearly all instruction sequences, as detailed in \cref{app:SA_results}.

The SA-based mapping also reveals the performance gaps between the outer and inner factory layouts, particularly for 2D layouts. To highlight the performance differences, we focus on the optimality gap, defined as the relative difference between the given output value and the optimal value. However, as it is difficult to directly compute the optimal value, we instead utilize a lower bound on the optimal solution, thereby obtaining an upper bound on the optimality gap. 
The lower bound on the optimal value is computed by ignoring magic state resource constraints and spatial congestion in the bus qubit area. In the terminology of \cref{sec:cbpi_stack}, this bound is equivalent to the sum of the base and the operand synchronization hazard.

The optimality gaps for the 2D layouts are at most 23.7\% and 11.3\% for the outer and inner factory layouts, respectively. This shows that the inner factory layout reduces the optimality gap by a factor of $2.1$. For the 2.5D layouts, the gaps are bounded from above by 7.9\% and 6.8\% in the outer and inner factory layouts; even at these near-optimal levels, the inner factory layout yields a further reduction. If the optimal solution is worse than the lower bound, the advantage of the inner factory layout becomes even more significant.

Note that Problem~\eqref{prob:mapping_optimization_problem} only approximates the placement quality and does not directly reflect the actual code beats in lattice surgery. In some cases, a smaller objective value may even yield worse performance after path computation. Nonetheless, SA solutions yield superior results in the vast majority of cases, and thus we consistently use SA-derived results in the main numerical experiments of \cref{sec:numerical_experiments}. The optimization of mapping is applicable to all architectures in this paper.

\section{Routing Optimization}
\label{sec:routing}

We next explain the routing problem.
We first summarize the spacetime routing technique~\cite{hamadaEfficientHighperformanceRouting2024} and describe its extension to 2.5D layouts in \cref{subsec:overview_3drouting}. We then point out that their main algorithm, which we call \textit{projective routing}, is generally incompatible with inner factory layouts in \cref{subsec:space_time_routing_limitation}, demonstrating the need for an efficient routing approach using a constant number of time slices. Finally, we detail the existing approaches and propose a novel routing method, \textit{double-slice routing}, in \cref{subsec:routing_methods}.

Our contribution to this routing optimization problem is twofold. First, we generalize the existing spacetime routing technique to accommodate 2.5D architectures (see \cref{thm:spacetime_25d}), thereby fully leveraging their architectural flexibility.
Second, we propose double-slice routing as a robust alternative to projective routing. By addressing the limitations of projective routing regarding measurement-based feedback, this approach facilitates seamless integration with inner factory layouts and enhances overall reliability through reduced circuit volume, while incurring only a negligible increase in execution time.

\subsection{Overview of Spacetime Routing} \label{subsec:overview_3drouting}

When every multi-body operation is decomposed into a sequence of two-body operations, the lattice-surgery scheduling problem can be simply framed as a 2D path packing problem on the qubit plane. A more sophisticated approach was proposed by \citet{beverlandSurfaceCodeCompilation2022}, who formulated the problem as vertex-disjoint path packing, which is then reduced to an edge-disjoint path packing problem. This results in an efficient scheduling algorithm with theoretical performance guarantees.

Following this approach, Ref.~\cite{hamadaEfficientHighperformanceRouting2024} provided an even more flexible routing scheme that leverages the temporal degrees of freedom. Concretely, routing can be formulated as a 3D path packing problem consisting of two spatial dimensions and one temporal dimension, subject to the following three conditions:
\begin{enumerate}[noitemsep]
    \item \emph{Path Condition}: The operand data patches must be connected by a 3D path through the bus patches.

    \item \emph{Boundary Condition}: The path must connect the operand data patches along the appropriate boundary types, as shown in \cref{tab:endpoint-boundary-types}.

    \item \emph{Kink Condition (informal)}: The parity of the number of \textit{kinks} on the path must be consistent with the operation type.
\end{enumerate}

\begin{table}[t]
    \centering
    \caption{Boundary-type requirements for each lattice-surgery operation}
    \label{tab:endpoint-boundary-types}
    \setlength{\tabcolsep}{5pt}
    \begin{tabular}{@{\,}lcc@{\,}}
        \toprule
        Operation       & Control boundary & Target boundary \\
        \midrule
        $XX$ measurement & $X$                & $X$               \\
        $ZZ$ measurement & $Z$                & $Z$               \\
        Horizontal move & $X$                & $X$               \\
        Vertical move   & $Z$                & $Z$               \\
        CNOT            & $Z$                & $X$               \\
        \bottomrule
    \end{tabular}
\end{table}

The first condition indicates that we allow only two data and/or factory patches for the joint logical measurement. The second condition is a desideratum due to the property of the surface code that the boundaries correspond to different logical operators. Note that these two conditions are also common when we formalize it as a 2D path packing problem.

Unique to the 3D formulation of the problem is the third condition. Here, a \textit{kink} is defined as a temporal segment that connects different boundary types at its ends (see \cref{fig:kink_table}). We assign either $XX$ or $ZZ$ measurements to each spatial segment. Then, a temporal segment can switch the measurement type if and only if it is a kink. Therefore, the boundary condition fixes the type of the spatial segment at the path endpoints, and the kink condition ensures that the measurement types of all the spatial segments can be determined consistently. Accordingly, the parity of kink numbers must be as follows:
\begin{enumerate}[noitemsep,start=3]
    \item \emph{Kink Condition}: An even number of kinks is required for $XX$ and $ZZ$ measurements and logical qubit movement, while an odd number is required for CNOT operations.
\end{enumerate}

\begin{figure}
    \centering
    \includegraphics[width=\linewidth]{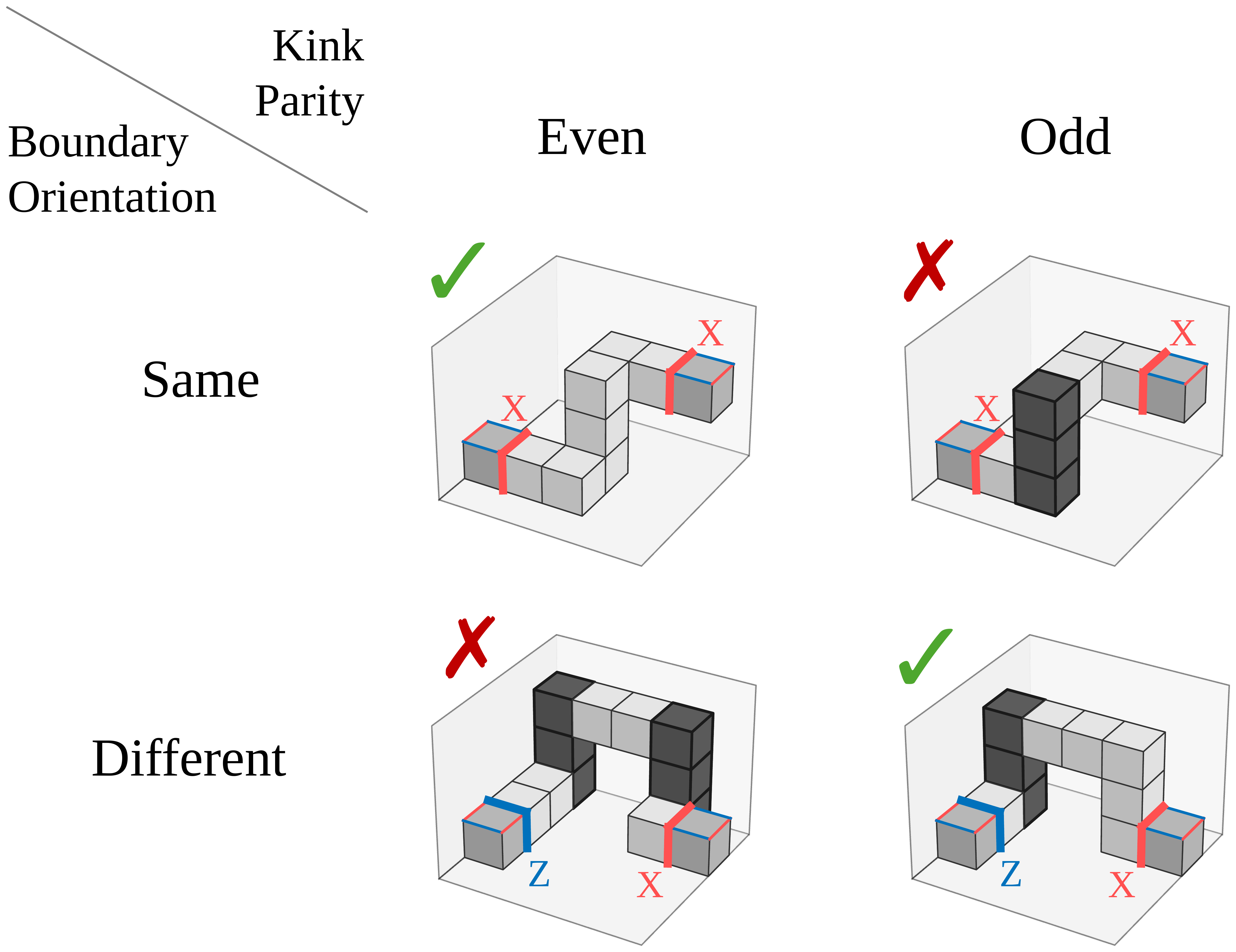}
    \caption{Valid and invalid 3D paths regarding the kink condition. Kinks are represented by dark voxels.}
    \label{fig:kink_table}
\end{figure}

Conversely, given a 3D path satisfying these three conditions, it can be verified that the measurements assigned to the spatial segments collectively constitute the desired action (see \cref{subsubsec:circuit_equivalence,app:spacetime_2d_simplification}; see also \cite[Section 3.3.2]{hamadaEfficientHighperformanceRouting2024}). Therefore, the conventional lattice-surgery routing problem can be generalized to the spacetime routing problem subject to the aforementioned three conditions.

In this paper, we further extend this spacetime routing framework~\cite[Theorem 1]{hamadaEfficientHighperformanceRouting2024} to 2.5D architectures. With only a modification to the definition of kinks, the same three conditions guarantee that a ``3.5D'' spacetime path can be explored to achieve the intended action. We formally state the extension below and defer its proof to \cref{subsubsec:spacetime_routing_extension_proof}.
\begin{restatable}{theorem}{spacetime}
    \label{thm:spacetime_25d}
    Suppose that the underlying 2.5D architecture supports either (i) inter-layer lattice surgery or (ii) transversal CNOT gates as inter-layer operations. Let the inter-layer movement of a path be treated analogously to a spatial segment for case (i) and a temporal segment for case (ii). Under this correspondence, kinks can be defined for a 3.5D spacetime path. Consequently, for any operation listed in \cref{tab:endpoint-boundary-types}, the path can achieve the intended operation if it satisfies the same three conditions.
\end{restatable}

\begin{figure*}[tb]
    \centering
    \includegraphics[width=\linewidth]{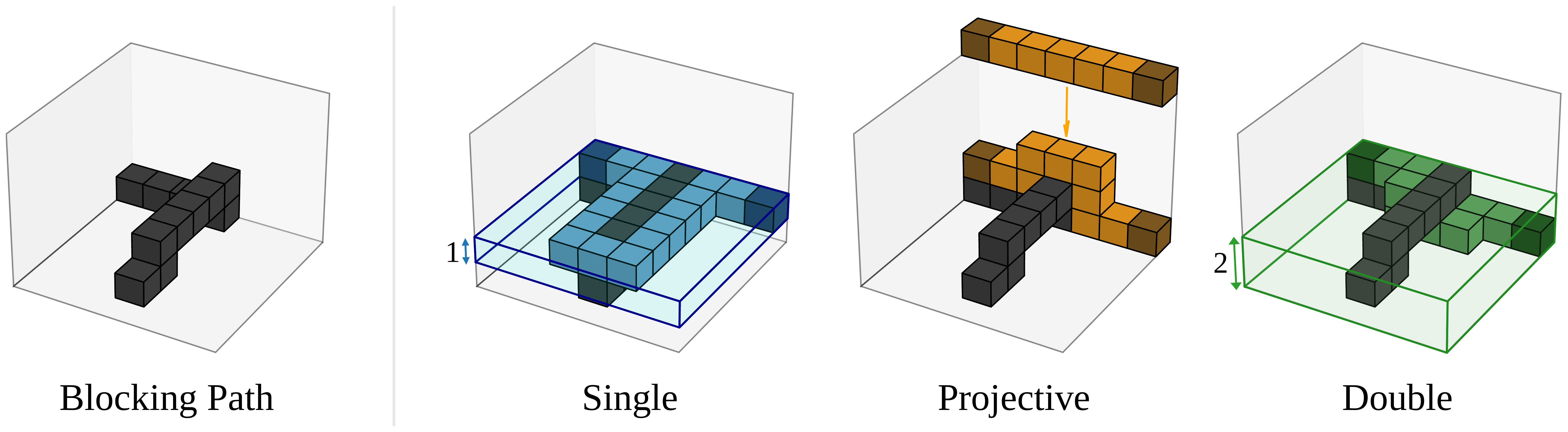}
    \caption{Spacetime visualization of three algorithms' example outputs.
    }
    \label{fig:routing_algo}
\end{figure*}

\subsection{Feedback Bottlenecks in Spacetime Routing}
\label{subsec:space_time_routing_limitation}
A critical limitation of spacetime routing emerges when integrating inner factory layouts with operations requiring measurement-based feedback, such as $T$-gate teleportation.

The standard $T$-gate protocol requires an $S$-gate correction based on a lattice-surgery measurement outcome. Since the measurement outcome of a lattice-surgery spacetime path cannot be finalized until all the measurement results of the spatial path segments are known, this implementation causes a severe temporal latency issue if the spacetime path spans numerous time slices. Ref.~\cite{hamadaEfficientHighperformanceRouting2024} argues that this latency can be mitigated by the use of auto-corrected $T$-gate teleportation~\cite{litinskiGameSurfaceCodes2019,Fowler2012TimeOptimal}, which converts the $S$-gate feedback into the choice of a Pauli measurement basis on a newly produced ancillary state adjacent to the original magic state location. This effectively translates the temporal latency issue into a spatial management problem: the ancillary logical state must be stored until the corresponding measurement outcome is determined.

The severity of this spatial bottleneck is fundamentally tied to the layout and routing framework. Outer factory layouts can readily remove these unresolved patches from the main computation area since they are produced at the perimeter of the qubit plane. In contrast, in inner factory layouts, these unresolved patches are generated within the primary computing region. Although one could move an unresolved patch outside the primary computation area using lattice surgery, this incurs essentially the same cost as moving a magic state from outer factories to the data qubit, thereby defeating the purpose of inner factory layouts.

Therefore, the extended path duration creates a severe blocking issue in inner factory layouts. As discussed in \cref{subsec:routing_methods}, the method proposed in Ref.~\cite{hamadaEfficientHighperformanceRouting2024} generates paths spanning an arbitrary number of time slices. Hence, the resulting unresolved patches act as long-term obstructions that congest the main computation area, rendering the method incompatible with such architectures. Given the promise of inner factories motivated by the magic state cultivation paradigm, this limitation highlights the urgent need for an efficient routing framework that operates with constant depth. This ensures that the ancillary states are resolved rapidly, keeping the qubit plane clear for subsequent operations.

\subsection{Routing Methods} \label{subsec:routing_methods}

\begin{table*}[t]
    \centering
    \caption{Comparison of routing methods by technical aspects. The column \textit{Kink Parity Correction} indicates the availability of a correction algorithm with guaranteed termination. Specifically, (--) denotes methods requiring two code beats per CNOT to avoid correction, while ($\dagger$) indicates that there exists a heuristic correction technique for 2.5D architectures with transversal CNOTs. The column \textit{Inner Layout Compatibility} refers to susceptibility to long-term obstructions discussed in \cref{subsec:space_time_routing_limitation}.}
    \label{tab:routing_technical_comparison}
    \newcommand{\ok}{\textcolor{green!60!black}{\checkmark}}
    \newcommand{\no}{\textcolor{red!80!black}{\ding{55}}}
    \begin{tabular*}{\textwidth}{@{\extracolsep{\fill}}lllccccc@{}}
        \toprule
        Layout & Method & & \# of Time Slices & Kink Parity Correction & Inner Factory Compatibility \\
        \midrule
        \multirow{3}{*}{\shortstack[l]{2D}} & Single & \cite{hamadaEfficientHighperformanceRouting2024} & 1 & -- & \ok \\
        & Projective & \cite{hamadaEfficientHighperformanceRouting2024} & $\infty$ & \ok & \no \\
        & Double & This work & 2 & \ok & \ok \\
        \midrule
        \multirow{3}{*}{\shortstack[l]{2.5D}} & Single & This work & 1 & --$\mathrlap{^\dagger}$ & \ok \\
        & Projective & This work & $\infty$ & \ok & \no \\
        & Double & This work & 2 & \ok & \ok \\
        \bottomrule
    \end{tabular*}
\end{table*}

In this section, we detail three routing algorithms: single-slice routing, projective routing, and double-slice routing. See \cref{fig:routing_algo} for their illustration, and \cref{tab:routing_technical_comparison} for their comparison from a technical aspect.

The first two methods were previously proposed in Ref.~\cite{hamadaEfficientHighperformanceRouting2024}. Single-slice routing (originally termed ``look-ahead BFS'', where ``BFS'' stands for Breadth-First Search) is a baseline algorithm without spacetime routing, and projective routing (originally termed ``look-ahead Dijkstra projection'') is a primary method leveraging the spacetime routing technique.

Building upon these approaches, we propose a novel routing method, double-slice routing, which utilizes only two time slices and thereby overcomes the feedback issue of projective routing discussed in \cref{subsec:space_time_routing_limitation}.

For each routing method, we first outline the routing strategy without considering the kink condition and then detail the correction of the kink parity. Further details and their extension to 2.5D architectures (including multiplanar configurations) are provided in \cref{app:routing}.

\subsubsection{Single-slice Routing}

\paragraph{Routing Strategy.}

First, we introduce single-slice routing as a baseline method. This framework does not exploit the spacetime routing technique and keeps each path within a single time slice.

The routing algorithm is a greedy approach with instruction look-ahead. For each time slice, it fetches instructions that have no unexecuted dependent instructions and finds the shortest path connecting their operand patches using BFS. We can efficiently retrieve available instructions by maintaining an instruction dependency graph and applying a technique similar to topological sorting (See \cite[Section 2.4]{hamadaEfficientHighperformanceRouting2024} for details). To optimize scheduling, instructions are prioritized based on their critical path length, which is defined as the maximum length of the dependency chain originating from the instruction.

\paragraph{Kink Parity Correction.}

Next, we provide a sketch of the strategies for addressing kink parity. Within a 2D layout constrained to a single time slice, it is inherently impossible to generate any kinks. This fundamental limitation prevents a direct implementation of CNOT operations via a path within a single time slice.

To accommodate CNOT operations in this framework, we employ the long-range CNOT approach~\cite{beverlandSurfaceCodeCompilation2022}, which effectively converts all corners into kinks while requiring two code beats per operation. Consequently, assuming that the data patches share the same boundary orientation, the kink condition is automatically satisfied by the boundary condition.

For 2.5D architectures supporting transversal CNOT gates, our implementation employs an efficient yet heuristic correction technique (see \cref{subsubsec:single_kink_25d}). Although this technique offers no guarantee of identifying a path that satisfies the kink condition, it performed robustly across all numerical benchmarks presented in this paper.

\subsubsection{Projective Routing}

\paragraph{Routing Strategy.}

Next, we introduce projective routing, a method that utilizes multiple time slices and incorporates projective search of paths~\cite{hamadaEfficientHighperformanceRouting2024}. In this framework, paths can span an arbitrary number of time slices. Since a naive path search in the 3D lattice would be time-consuming, the algorithm constructs 3D paths by stacking 2D paths onto the current scheduling result. To create the intermediate 2D path, the algorithm searches a weighted grid for a minimum-weight path, where patch weights increase as powers of two relative to their current height in the scheduling result.

Because the routing procedure always succeeds in finding a path, a straightforward application of the instruction look-ahead technique, which is used for single-slice routing, is ineffective. Instead, the algorithm employs an alternative look-ahead strategy as an optimization: it repeatedly fetches an instruction that minimizes the maximum height of the operand patches.

As later numerical experiments show, projective routing outputs the smallest execution time among the three methods discussed in this section, but its circuit volume is relatively large. We expect this method to be beneficial for long-term FTQC architectures, where minimizing the execution time is critical.

\paragraph{Kink Parity Correction.}

Since 3D paths are obtained by projection, the easiest adjustment to these paths would be to lift the path voxels. We implement a kink correction technique presented in Ref.~\cite{hamadaEfficientHighperformanceRouting2024}, which identifies the first corner of a path and locally lifts voxels to flip the kink parity. This technique is guaranteed to correctly adjust the parity, provided that the data qubits share the same boundary orientation. In this work, we extended this correction technique to 2.5D architectures. See \cref{subsubsec:proj_kink_25d} for further details.

\subsubsection{Double-Slice Routing}

\paragraph{Routing Strategy.}

Finally, we propose a new method called double-slice routing. Similar to the single-slice routing, this algorithm is also a greedy approach with instruction look-ahead. The major difference from the single-slice routing is that this framework utilizes 3D routing, specifically focusing on paths that span two time slices. The algorithm manages the two most recent slices and finds the shortest path within them. The use of two time slices is effective because it prevents routing collisions while maintaining short lattice surgery paths, combining the advantages of both single-slice and projective routing. Furthermore, it enables CNOT routing in a single path search and facilitates the implementation of a kink correction subroutine.

As our later numerical experiments show, although double-slice routing results in a slightly longer execution time than projective routing, it achieves shorter path lengths and a smaller circuit volume. We believe this characteristic makes the method particularly valuable for the early FTQC architectures, where error suppression and decoding efficiency are critical. 

\paragraph{Kink Parity Correction.}

As a kink parity correction method tailored to this framework, we propose a technique that ``pinches'' two spatially adjacent voxels in the time direction (see \cref{fig:double_kink_correction}). To correct kink parity, the algorithm exhaustively examines every pair of spatially adjacent voxels, shifting them temporally within the two time slices and updating the path accordingly. Notably, this approach requires only two free voxels adjacent to the path, which is the optimal constant upper bound for spatial overhead. This is justified by the fact that even a 2D path with a single corner demands at least two free additional voxels to form a kink.

\begin{figure*}[tb]
    \begin{subfigure}
        [t]{0.3\linewidth}
        \centering
        \includegraphics[width=\columnwidth]{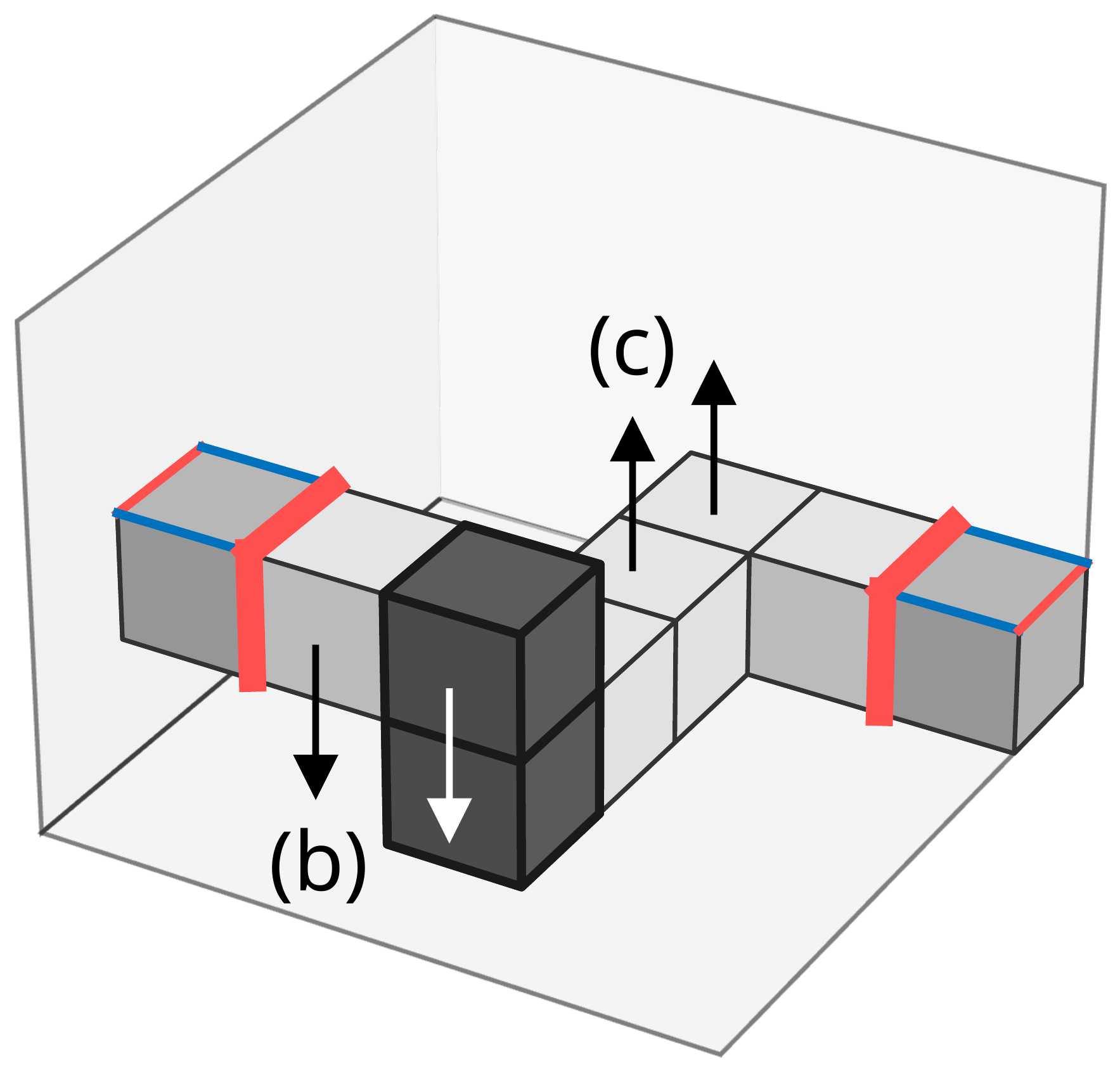}
        \subcaption{}
    \end{subfigure}
    \hfill
    \begin{subfigure}
        [t]{0.3\linewidth}
        \centering
        \includegraphics[width=\columnwidth]{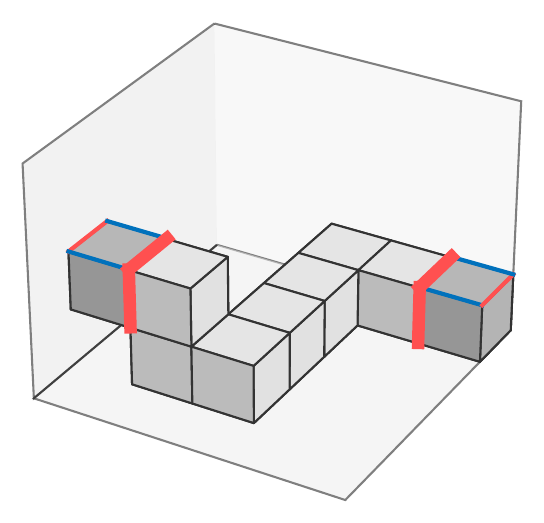}
        \subcaption{}
    \end{subfigure}
    \hfill
    \begin{subfigure}
        [t]{0.3\linewidth}
        \centering
        \includegraphics[width=\columnwidth]{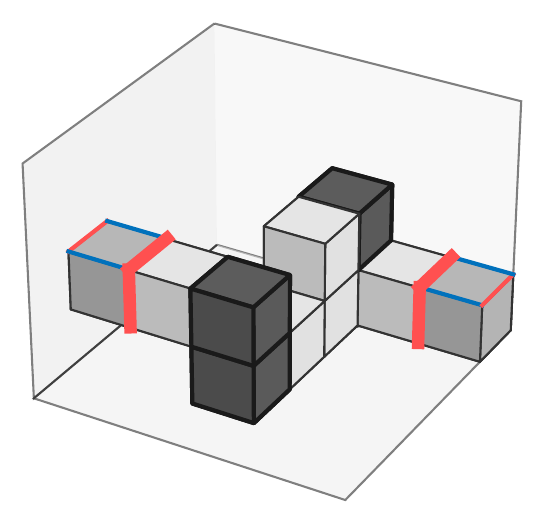}
        \subcaption{}
    \end{subfigure}
    \caption{Illustration of the kink parity correction technique for double-slice routing. Kinks are represented by dark pillars of voxels. (a) A path spanning two time slices before kink correction. (b, c) Corrected paths after temporally pinching two spatially adjacent voxels of the path, which now satisfy the kink condition.}
    \label{fig:double_kink_correction}
\end{figure*}

A potential drawback of this adjustment is its susceptibility to failure when temporally adjacent voxels are occupied. If the algorithm fails to correct the parity, it abandons the current path and proceeds to the next executable instruction, which is exactly the same behavior as when it fails to find a path between the operands. Nevertheless, provided that all data qubits share a common boundary orientation, the entire algorithm is guaranteed to terminate successfully. To this end, we prove in \cref{subsec:termination_guarantee} that our adjustment procedure successfully corrects the kink parity of a spacetime path under three assumptions: (i) the two target qubits have the same boundary orientation, (ii) none of the temporally adjacent voxels in the two focused time slices is blocked by other paths, and (iii) the path before correction is the shortest path. While the first and third conditions are always satisfied by our assumption on boundary orientations and the design of the path search algorithm, respectively, not all routing attempts initially satisfy the second condition. However, if the algorithm repeatedly fails to find a path or correct the kink parity for a given instruction, all preceding instructions will eventually be resolved; at this point, the absence of other blocking paths ensures the second condition and allows the routing path to be successfully corrected.

\section{Numerical Experiments}
\label{sec:numerical_experiments}

\subsection{CBPI Stack}
\label{sec:cbpi_stack}

To evaluate the performance of our compilation technique, we use the metric called Code Beats Per Instruction (CBPI)~\cite{uenoHighPerformanceScalableFaultTolerant2024}.
By dividing the total execution time by the number of instructions, the CBPI quantifies the instruction-level parallelism and the efficiency of the lattice surgery operations independently of the length of the input program.

Furthermore, we utilize the CBPI stack~\cite{uenoHighPerformanceScalableFaultTolerant2024} to analyze factors that degrade performance, called \textit{hazards}. This study focuses on four hazards: operand synchronization, CX path congestion, magic path congestion, and kink parity correction. We first establish a hazard-free baseline performance and then sequentially introduce these constraints to attribute the resulting performance loss to each specific hazard. In this paper, we apply this decomposition technique to both total execution time and total circuit volume.

The following outlines the computation for the base and each hazard (see \cref{sec:cbpi_detail} for further details).
\begin{itemize}
    \item \textit{Base}: This represents the theoretical lower bound of execution time, assuming that the instructions exclusively require one code beat on their operand qubits. Hence, this ideal execution time equals the maximum number of instructions applied to any single qubit.
    \item \textit{Operand Synchronization}: This hazard arises from the timing constraints inherent to the routing framework, such as single- or double-slice routing. This component, combined with the base, represents the performance limit for a given framework.
    \item \textit{CX Path Congestion}: This component accounts for routing path congestion caused by all \texttt{CX} instructions.
    \item \textit{Magic Path Congestion}: This hazard reflects the finite throughput of the MSF patches and spatial congestion of their incident routing paths.
    \item \textit{Kink Parity Correction}: This represents the final overhead required to modify spacetime paths to satisfy the kink condition.
\end{itemize}

\subsection{Results}

\subsubsection{Routing Comparison}

\begin{figure*}
    \centering
    \includegraphics[width=1\linewidth]{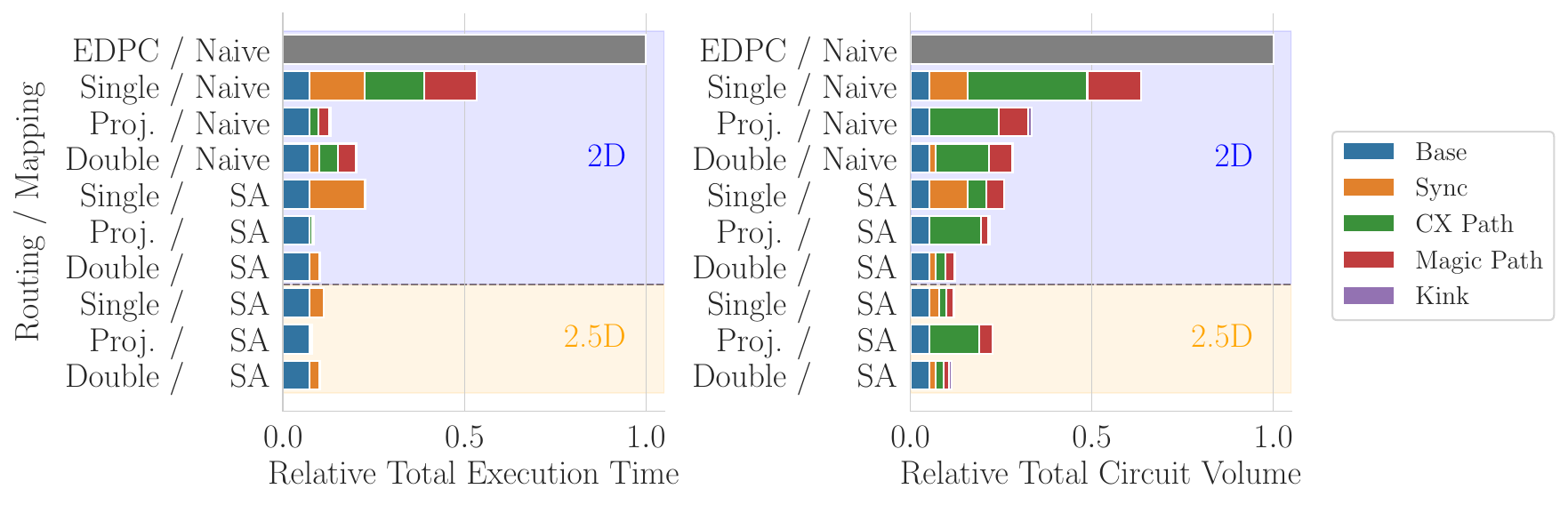}
    \caption{Benchmarks of the three routing methods for the SELECT-5 circuit with an outer architecture, evaluated on the two metrics. All the values are normalized relative to the EDPC baseline for better clarity. }
    \label{fig:routing_comparison}
\end{figure*}

This section presents several numerical experimental results with a focus on the performance trends of the three routing methods. We also included the Edge-Disjoint Paths Compilation (EDPC) algorithm by \citet{beverlandSurfaceCodeCompilation2022} as an existing method for large-scale compilation.

\begin{table*}[t]
    \centering
    \caption{Performance comparison across different allocation and routing methods. The values represent the improvement ratios between the specified configurations. Routing methods are denoted by their initial letters: S, D, and P stand for single-slice, double-slice, and projective routing, respectively.}
    \label{tab:performance_comparison}
    \renewcommand{~}{\hphantom{0}}
    \begin{tabular*}{\textwidth}{@{\extracolsep{\fill}}llcccc@{}}
        \toprule
        \textbf{Metric} & \textbf{Benchmark} & \textbf{SA Gain} & \textbf{Routing Gain} & \textbf{Total Gain} & \textbf{Ratio to Proj.} \\
        & & (Naive-S / SA-S) & (SA-S / SA-D) & (Naive-S / SA-D) & (SA-P / SA-D) \\
        \midrule
        \multirow{3}{*}{\shortstack[l]{Execution Time}}
        & SELECT-4 & 1.9 & 2.2 & 4.1 & ~0.85 \\
        & SELECT-5 & 2.4 & 2.2 & 5.2 & ~0.83 \\
        & SELECT-6 & 3.2 & 2.4 & 7.5 & ~0.77 \\
        \midrule
        \multirow{3}{*}{\shortstack[l]{Circuit Volume}}
        & SELECT-4 & 2.3 & 2.0 & 4.7 & 2.2 \\
        & SELECT-5 & 2.5 & 2.1 & 5.1 & 1.8 \\
        & SELECT-6 & 3.0 & 2.1 & 6.2 & 1.3 \\
        \bottomrule
    \end{tabular*}
\end{table*}

\Cref{fig:routing_comparison} summarizes the performance trend of the routing algorithms benchmarked on the SELECT-5 circuit. To align with the experimental settings of the EDPC algorithm, we adopted the 2D outer factory layout and set $\tau$ to $0$. Since the EDPC algorithm uses the naive allocation, we included both the naive allocation and the SA allocation. All the values are normalized relative to the EDPC baseline for better visibility. 

In this comparison, the baseline algorithm already outperforms the EDPC algorithm. This advantage is due to a fundamental difference in managing instruction dependency. The EDPC algorithm partitions the instruction dependency graph into sequential layers via a topological sort and then processes them in order. In contrast, the instruction look-ahead technique employed in single-slice and double-slice routing dynamically identifies and attempts to route all executable instructions from the entire graph at each step, offering greater flexibility.

This distinction directly impacts performance. With EDPC, a single graph layer often splits into multiple time slices due to routing conflicts, which frequently results in sparse, underutilized slices. These findings suggest that this scheduling inefficiency, rather than a difference in the routing logic itself, is the primary source of the performance gap.

We note that the layered approach by EDPC is well-suited for handling a wider set of operations, such as Hadamard gates, which our method currently ignores. It also offers a clear advantage for parallel execution by partitioning the entire problem across graph layers.

We now compare the three routing methods described in \cref{sec:routing}. \Cref{tab:performance_comparison} shows the improvement ratios regarding these methods on the parallelized SELECT circuits.
For most of the tested circuits, projective routing yields the fewest code beats, followed by double-slice routing and single-slice routing in this order.
This is a natural consequence of path flexibility in each framework. Regarding the circuit volumes, double-slice routing outputs the smallest volume.

On the other hand, double-slice routing and projective routing output worse scheduling results than single-slice routing in a few cases. The performance degradation in double-slice routing can be attributed to a non-negligible kink correction overhead. Projective routing has the problem of blocking the qubit plane by temporal segments when the plane is small, resulting in significant delays for subsequent routing.

\subsubsection{Distribution of Path Volumes}

To investigate the behavior of the three routing algorithms discussed in \cref{subsec:routing_methods}, we analyze the distribution of \textit{path volumes} in the scheduling results of each method. Here, a path volume of $P$ refers to the number of spacetime pairs $(x, t) \in P \subseteq X \times \mathbb{N}$ contained in the 3D path.

\begin{figure}[tb]
    \centering
    \includegraphics[width=\columnwidth]{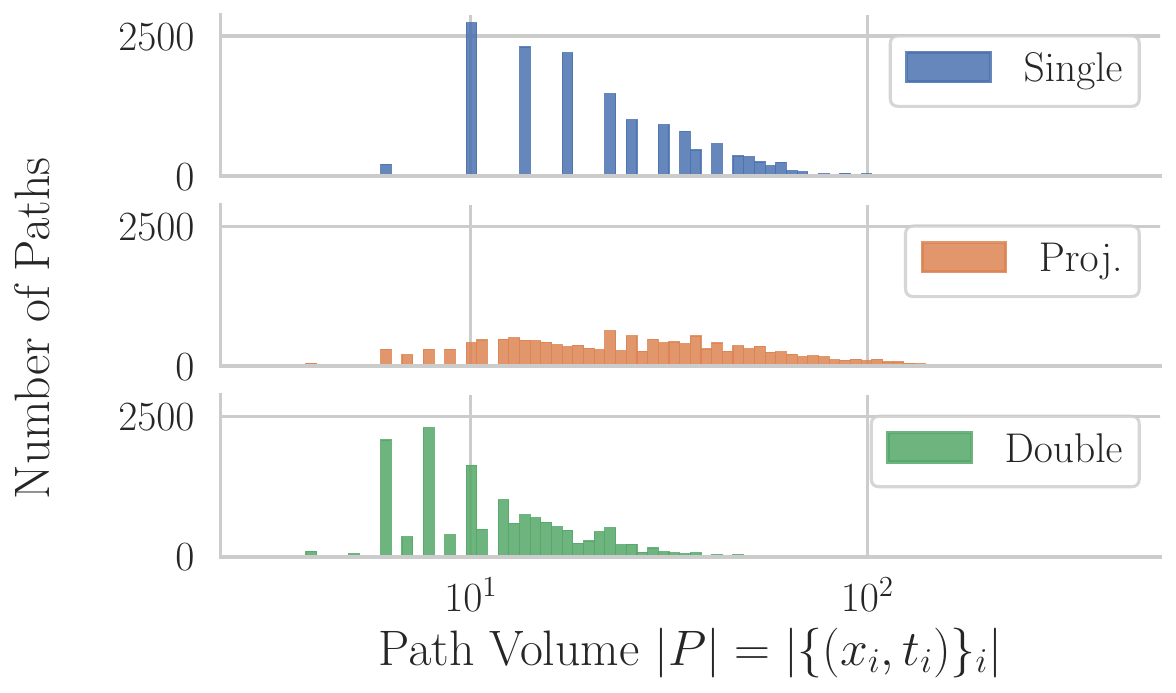}
    \caption{Histograms of path volumes for the three routing algorithms evaluated on the SELECT-5 circuit simulating the 2D Heisenberg model. The x-axis is shown on a logarithmic scale.}
    \label{fig:path_length_histogram}
\end{figure}

Following the benchmark configuration from the previous section, we evaluate the routing methods on the 2D outer factory layout optimized with the SA-based mapping technique, while setting $\tau$ to $0$ and using the SELECT-5 circuit simulating the 2D Heisenberg model.
\Cref{fig:path_length_histogram} shows the histograms for the three scheduling results.
The distribution for double-slice routing is shifted to the left compared to the other two, indicating that the method reliably yields short paths.

\begin{figure}[tb]
    \centering
    \includegraphics[width=\columnwidth]{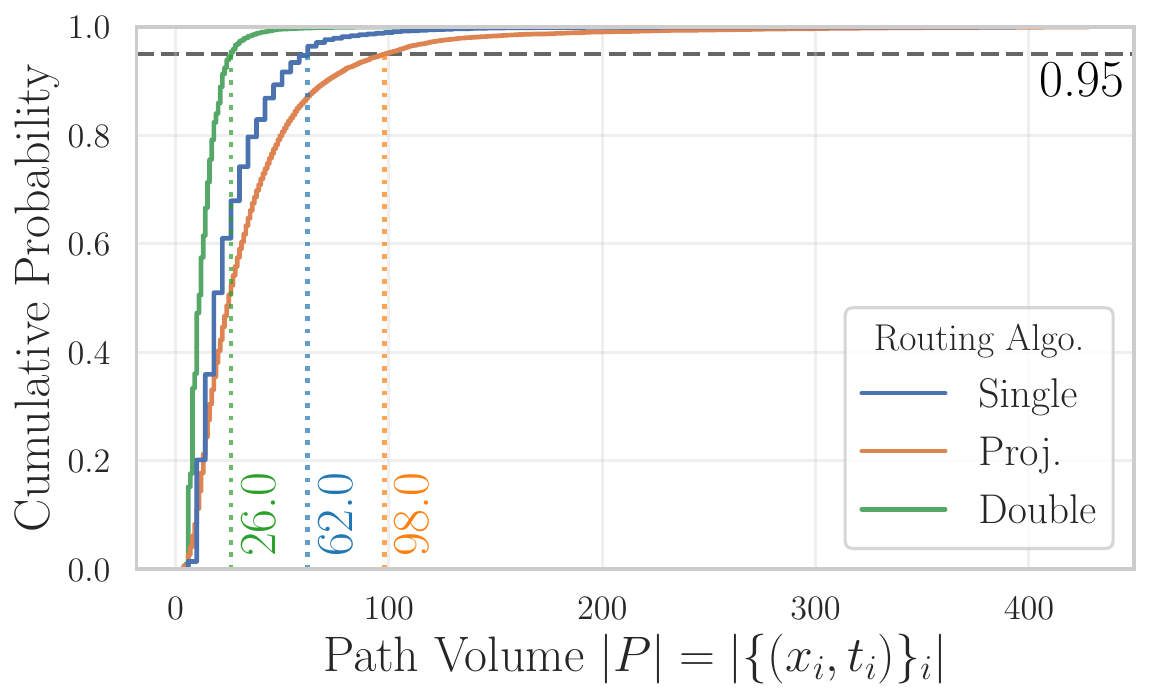}
    \caption{Empirical Cumulative Distribution Function (ECDF) of path volumes for the three routing algorithms evaluated on the SELECT-5 circuit simulating the 2D Heisenberg model. Dotted lines indicate the 95th percentiles, representing the tail performance.}
    \label{fig:path_length_ecdf}
\end{figure}

This robustness is further verified by the Empirical Cumulative Distribution Function (ECDF) shown in \cref{fig:path_length_ecdf}. The 95th percentiles suggest that double-slice routing consistently provides short paths, whereas the tail performance of projective routing is significantly inferior. Notably, the maximum path length for double-slice routing is $146$, whereas those for single-slice routing and projective routing are $394$ and $428$, respectively.

The consistent acquisition of shorter paths, combined with the overall reduction in total circuit volume, strongly suggests a significant improvement in decoding latency. By mitigating extreme path lengths, double-slice routing suppresses potential bottlenecks in real-time error correction. However, as the actual decoding latency depends on specific algorithmic implementations and windowing strategies, a rigorous evaluation remains a subject for future work.

\subsubsection{Performance Evaluation on SELECT and Trotter Circuits}
\label{subsubsec:select_trotter}

\begin{figure*}
    \centering
    \includegraphics[width=\linewidth]{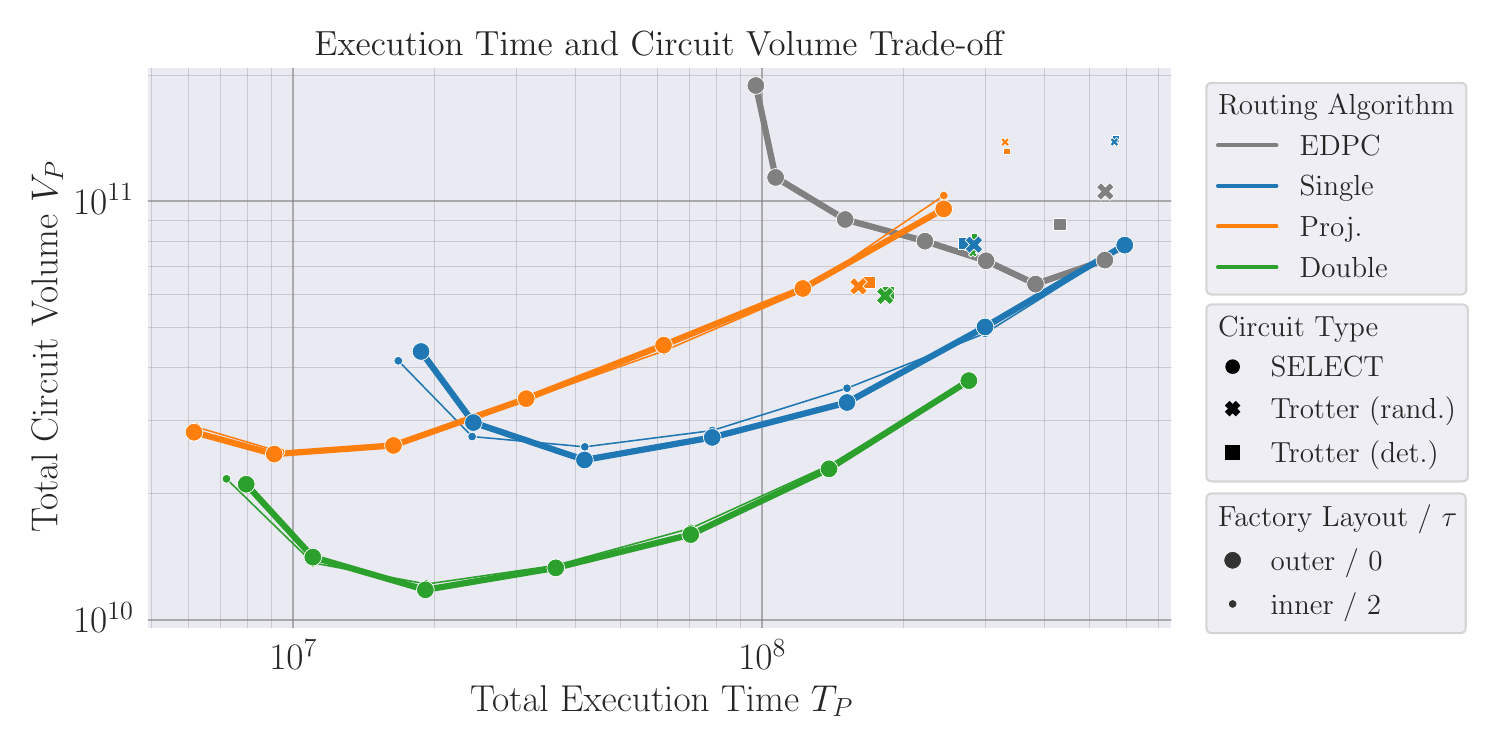}
    \caption{Scheduling performance of each method evaluated on the circuits simulating the 2D Heisenberg model. Each color represents a different routing method, with each point corresponding to a specific circuit. All results were generated using the 2D layout.
    }
    \label{fig:cv_cb_plot}
\end{figure*}

\Cref{fig:cv_cb_plot} presents a performance comparison of different scheduling methods and circuits for the 2D Heisenberg model using the 2D layout. In this comparison, we also included the comparison of factory layouts: the outer layout with $\tau = 0$ and the inner layout with $\tau = 2$. These $\tau$ values are chosen as reasonable values in each factory layout, while a more detailed comparison is deferred to \cref{sec:breakdown_outer_inner}. In this plot, the circuit volume is graphed against the required code beats. The lines in the plot correspond to the SELECT-$i$ circuits with different parallelism levels ranging from $i = 0$ to $6$. The rightmost and leftmost points correspond to $i = 0$ and $i = 6$, respectively.

Regarding the SELECT circuits, we observed that the $2^{i}$-thread parallelization effectively reduces the FTQC runtime up to the tested limit of $i = 6$, while the circuit volume is minimized at $i = 4$ for single- and double-slice routing, and at $i = 5$ for projective routing.

As a comparison of Trotter circuits and SELECT circuits, it is observed that, for each routing method, the SELECT circuits with an appropriate level of parallelization (e.g., $i=4$) outperform the Trotter circuits in both metrics.

Regarding the comparison of the factory layouts, despite the reduced throughput of the MSF patches in the inner layout with $\tau = 2$, its performance remains comparable to that of the outer layouts with $\tau = 0$ on the SELECT circuits. This demonstrates that mitigating path collisions involving magic states offers substantial benefits, effectively compensating for the throughput loss. Conversely, this advantage is slight within the small qubit plane used for the Trotter circuit, resulting in a significant performance degradation when employing the inner layout with $\tau = 2$.

\section{Discussion}
\label{sec:discussion}

In this work, we developed a compilation framework for lattice-surgery-based FTQC that jointly optimizes mapping and routing. We formulated a mapping problem that treats data patches, routing ancillae, and MSF patches on equal footing, thereby enabling inner factory layouts motivated by recent advances in magic state cultivation. We also extended spacetime routing to 2.5D architectures and introduced double-slice routing, a constant-depth alternative to projective routing equipped with a kink-parity correction procedure that terminates under mild assumptions. Benchmarks on Hamiltonian-simulation circuits showed that mapping optimization improves execution time by up to a factor of 3.2 over baseline placements, while double-slice routing reduces circuit volume by up to a factor of 2.4 relative to single-slice routing. These results identify joint optimization of mapping and routing as a practical way to navigate the trade-off between execution time and circuit volume in lattice-surgery compilation.

We envision numerous intriguing future works.
The first is decoder-aware compilation. In the present work, we effectively assume that routing paths can be designed with substantial geometric freedom, so that the compiler optimizes spacetime cost primarily at the level of path packing and scheduling. In a realistic implementation, however, this freedom may be limited by the requirements of real-time decoding. Constructing and validating a dedicated decoder for each dynamically generated routing pattern is likely to be impractical, whereas a more scalable strategy is to precompile and cache a restricted library of decoder-compatible routing motifs. This would impose explicit shape constraints on admissible routes and may substantially change the compilation landscape. An important open problem is therefore to develop compilation methods that minimize spacetime volume while restricting routing to a decoder-friendly family of precompiled path geometries.

Another important direction is to extend spacetime routing beyond pairwise lattice surgery to many-body lattice-surgery primitives. Such an extension appears relevant for approaches that exploit global Pauli rotations, as in Pauli-based computation, or that preserve parallelism by partially applying Clifford conjugations to reduce the support of multi-Pauli measurements. In these settings, the routing object generalizes from a simple path between two operands to a complex tree branching across multiple operands. Given that Ref.~\cite[Theorem 2]{hamadaEfficientHighperformanceRouting2024} has recently formulated the necessary structural conditions and extended projective routing to multi-body cases, integrating our routing algorithm into their theoretical formulation presents a promising direction for future work. Such a synthesis would yield an efficient compilation algorithm for these generalized objects that is also compatible with inner factory layouts.

\section*{Acknowledgments}

N.Y.\ is supported by JST Grant Number JPMJPF2221, JST CREST Grant Number JPMJCR23I4, IBM Quantum, Google Quantum AI, JST ASPIRE Grant Number JPMJAP2316, JST ERATO Grant Number JPMJER2302, and the Institute of AI and Beyond of the University of Tokyo, JST [Moonshot R\&D] [Grant Number JPMJMS256J]. 
Y.S.\ is supported by MEXT Q-LEAP Grant Number JPMXS0120319794 and JPMXS0118068682, MEXT Feasibility Study on the future HPCI, JST [Moonshot R\&D] [Grant Number JPMJMS2061], JST CREST Grant Number JPMJCR24I4 and JPMJCR25I4.
Y.U.\ is supported by JSPS KAKENHI Grant Number JP25K21176.

\bibliographystyle{apsrev4-2}
\bibliography{zotero_generated.bib}

\clearpage
\appendix
\onecolumngrid

\section{Target quantum circuit}
\label{app:target}

In this Appendix, we provide details on the target quantum circuits evaluated in the main text.
We first explain the quantum algorithms considered in this work in \cref{subsec:quantum_algorithm}, and then describe the benchmark Hamiltonians in \cref{subsec:hamiltonian}.

\subsection{Quantum algorithm} \label{subsec:quantum_algorithm}

\subsubsection{Quantum singular value transformation}

Quantum singular value transformation (QSVT) is a general framework that implements polynomial transformations of the singular values of a block-encoded operator through a sequence of signal-processing phase rotations.
Owing to its broad applicability to Hamiltonian simulation, linear-system solvers, ground-state preparation, quantum machine learning, and related tasks, QSVT is widely regarded as one of the central algorithmic primitives for fault-tolerant quantum computing.

An essential bottleneck in QSVT-based Hamiltonian simulation is the implementation of the {\tt SELECT} operator, which is combined with the {\tt PREPARE} operator to realize a block encoding of a target Hamiltonian.
For many lattice Hamiltonians of practical interest, the cost of {\tt PREPARE} is relatively modest, whereas {\tt SELECT} contains a large number of controlled Pauli operations and therefore dominates the logical gate count.
For this reason, the main text focuses on the compilation cost of the {\tt SELECT} circuit.

Concretely, given a decomposition of an $n$-qubit Hamiltonian $H$ into a weighted sum of $L$ Pauli operators as
\begin{equation*}
    H = \sum_{i=1}^{L}\alpha_{i}P_{i}\quad (\alpha_{i}\geq 0, \; P_{i}\in \{\pm1\}\times\{I,X,Y,Z\}^{\otimes n}).
\end{equation*}
We define the SELECT and PREPARE operators as unitaries that act on a composite system of the target space and $a=\lceil \log_2 L \rceil$ ancillary qubit space as follows,
\begin{equation*}
    U_{\mathrm{SELECT}}\defeq \sum_{i=1}^{L}\ketbra{i}{i}\otimes P_{i}, \quad U_{\mathrm{PREPARE}}\ket{0}^{\otimes a}= \sum_{i=1}^{L}\sqrt{\frac{\alpha_{i}}{\lambda}}\ket{i},
\end{equation*}
where $\lambda \defeq \sum_{i=1}^{L}\alpha_{i}$.
One can show that $^{\otimes a}\langle 0 | U_{\rm PREPARE}^\dagger U_{\rm SELECT} U_{\rm PREPARE} | 0\rangle^{\otimes a} = H / \lambda$, indicating the block encoding of the renormalized Hamiltonian.

For lattice Hamiltonians with translationally invariant interaction profiles such as the Heisenberg model and the Fermi--Hubbard model, the PREPARE circuit requires gate complexity of $O(\log n/\epsilon)$ with $\epsilon$ indicating the Clifford+T synthesis accuracy, whereas the SELECT circuit requires a significantly larger gate count of $O(n)$. Indeed, this results in a significant separation in gate count for practical ground state energy estimation problems, and thus it is common to analyze the cost of a single SELECT circuit.
In this context, Ref.~\cite{yoshiokaHuntingQuantumclassicalCrossover2024} proposed a parallelized version of the SELECT circuit that shortens the circuit execution time.
Let $M$ be a power of two less than $L$. The parallelized SELECT circuit divides the  operation into $M$ sub-circuits, each utilizing $\log(L/M)$ index qubits and $\log(L/M)$ ancilla qubits. For moderate $M$, the runtime is reduced by nearly a factor of $M$. This speedup saturates either when the magic state supply rate cannot keep pace with parallel gate execution, or when
$M$ becomes so large that there is insufficient work per thread to exploit further parallelism.

\subsubsection{Trotter Circuits}

Another important quantum algorithm for Hamiltonian simulation is the Trotterization. Here, we exclusively consider the second-order Trotterization.
Assume that we decompose the Hamiltonian as $H=\sum_{j=1}^{m}H_{j}$ where each $H_j$ is composed of mutually commuting Pauli terms. Choosing step size as $\tau=t/r$ with $r\in \mathbb{N}$, the second-order Suzuki--Trotter formula yields
\begin{equation*}
    e^{-iHt} \approx \qty(\prod_{j=1}^{m}e^{-iH_j \tau/2} \prod_{j=m}^{1} e^{-i H_j \tau/2})^{r}.
\end{equation*}
The error between the exact unitary $U = e^{-i H t}$ and the Trotterized unitary $U_{\rm TS}$ can be upper bounded as
\begin{equation*}
    \|U - U_{\rm TS}\| \leq r\frac{\tau^3}{12} \sum_{b = 1}^m \left(
    \left\| \sum_{c > b} \sum_{a > b} [[H_b, H_c], H_a] \right\| + \frac{1}{2} \left \| \sum_{c > b}[[H_b, H_c], H_b] \right \|
    \right) \eqcolon W r \tau^3,
\end{equation*}
where $W$ is further upper bounded based on the Pauli coefficients of terms in $\sum_{c>b} \sum_{a >b}[[H_b, H_c], H_a]$ and $ \sum_{c > b} [[H_b, H_c], H_b]$.

\begin{table}[t]
    \centering
    \caption{Hamiltonians for SELECT and Trotter circuits.}
    \label{tab:Hamiltonians}

    \begin{minipage}{0.48\linewidth}
        \centering
        \textbf{For SELECT circuits}

        \vspace{0.5em}
        \begin{tabular}{lcc}
            \toprule
            name             & size & boundary \\
            \midrule
            FermiHubbard2D   & 10   & cylinder \\
            Heisenberg2D     & 10   & cylinder \\
            Z2LatticeGauge2D & 10   & PBC      \\
            Heisenberg1D     & 100  & OBC      \\
            Random Local     & 100  & OBC      \\
            Schwinger        & 100  & OBC      \\
            \bottomrule
        \end{tabular}
    \end{minipage}
    \hfill
    \begin{minipage}{0.48\linewidth}
        \centering
        \textbf{For Trotter circuits}

        \vspace{0.5em}
        \begin{tabular}{lcc}
            \toprule
            name           & size & boundary \\
            \midrule
            FermiHubbard2D & 10   & cylinder \\
            Random Local   & 10   & OBC      \\
            Schwinger      & 20   & OBC      \\
            Heisenberg2D   & 30   & cylinder \\
            \bottomrule
        \end{tabular}
    \end{minipage}
\end{table}

\subsection{Hamiltonian}\label{subsec:hamiltonian}

We have benchmarked the cost of both SELECT and Trotter circuits for various models as shown in \cref{tab:Hamiltonians}.
While we have only shown the results for the 2D Heisenberg Hamiltonian in the main text, we provide the full data in our GitHub repository~\cite{github}.
In the following, we briefly provide the definition of each model.

\subsubsection{Heisenberg model}

The Heisenberg model describes two-body interactions among quantum spin operators
$\{\hat{S}_{i}^{(a)}\}_{i\in \Lambda, a\in \{X, Y, Z\}}$ representing localized degrees of freedom on a lattice $\Lambda$:
\begin{equation}
    H = \sum_{\langle i, j \rangle \in \Lambda}\left(S_{i}^{x}S_{j}^{x}+ S_{i}^{y}S_{j}^{y}+ S_{i}^{z}S_{j}^{z}\right).
\end{equation}
Here, $\langle i, j \rangle$ denotes a pair of sites connected by an edge of the lattice.
By allowing interaction between sites at Manhattan distance two on $\Lambda$, we can define the $J_{1}$--$J_{2}$ Heisenberg model as
\begin{equation}
    H
    = J_{1}\sum_{\langle i, j\rangle}\left(S_{i}^{x}S_{j}^{x}+ S_{i}^{y}S_{j}^{y}+ S_{i}^{z}S_{j}^{z}\right)
    + J_{2}\sum_{\langle\langle i, j\rangle\rangle}\left(S_{i}^{x}S_{j}^{x}+ S_{i}^{y}S_{j}^{y}+ S_{i}^{z}S_{j}^{z}\right),
\end{equation}
where $\langle\langle \cdot \rangle \rangle$ indicates the summation over pairs of next-nearest sites on the square lattice.
Here, $J_{1}$ and $J_{2}$ represent the interaction strengths between pairs of sites at Manhattan distance 1 and 2, respectively.

While the physical properties of the model depend on the magnitudes of the coupling terms, it is clear that the complexity of the {\tt SELECT} operation is determined by the interaction graph $G=(V,E)$ specifying whether an interaction term is present.
In particular, when each quantum spin has angular momentum $1/2$, the operators $S_{i}^{x,y,z}$ can be identified with Pauli operators (up to an overall constant factor), i.e., $2S_i^{x}\equiv X_i$, $2S_i^{y}\equiv Y_i$, and $2S_i^{z}\equiv Z_i$.

In the main text, we have considered the $J_1$-$J_2$ Heisenberg model of spin 1/2 on the square lattice with $J_1 = 1, J_2 = 0.5$, which is denoted as ``Heisenberg2D'' in \cref{tab:Hamiltonians}.
Note that ``Heisenberg1D'' rather refers to the model with spin 1 with only nearest neighbor interaction on a chain.

\subsubsection{Fermi--Hubbard Hamiltonian}
Among Fermi--Hubbard-type models that include two-body interaction terms,
the most general form can be written as
\begin{equation}
    H = \sum_{i j}t_{i, j}c_{i}^{\dagger}c_{j}+ \sum_{ijkl}V_{ijkl}c_{i}^{\dagger}c_{j}^{\dagger}c_{k}c_{l}.
\end{equation}
Here, $c_{i}^{(\dagger)}$ is the annihilation (creation) operator for the $i$-th fermionic mode, satisfying the anticommutation relation
$\{c_{i}, c_{j}^{\dagger}\} = c_{i}c_{j}^{\dagger}+ c_{j}^{\dagger}c_{i}= \delta_{ij}$.
Moreover, $t_{ij}$ represents hopping between modes $i$ and $j$, and $V_{ijkl}$ specifies the interaction strength.
For first-principles Hamiltonians describing electronic states in molecules and solids, $t_{ij}$ and $V_{ijkl}$ are typically dense, whereas in condensed-matter physics, both $t$ and $V$ are often specified sparsely.
In particular, one often associates the Hamiltonian with a graph $G=(V,E)$ and considers the following form:
\begin{equation}
    H = (-t)\sum_{(i, j)\in E}\sum_{\sigma \in \{\uparrow, \downarrow\}}(c_{i,\sigma}^{\dagger}c_{j, \sigma}+{\rm h.c.})
    + U\sum_{i\in V}c_{i, \uparrow}^{\dagger}c_{i,\uparrow}c_{i, \downarrow}^{\dagger}c_{i, \downarrow}
    - \mu \sum_{i\in V}\sum_{\sigma \in \{\uparrow, \downarrow\}}c_{i, \sigma}^{\dag}c_{i, \sigma}.
    \label{eq:fermi-hubbard}
\end{equation}
Here, $\sigma \in \{\uparrow, \downarrow\}$ denotes the spin quantum number. The Hubbard model was originally introduced to describe electronic states, reflecting the fact that electrons are spin-$1/2$ particles.

To simulate fermionic systems on a quantum computer, one needs to translate fermionic operators into qubit operators. Several such mappings are known; here we focus on the widely used Jordan--Wigner transformation. Concretely, the correspondence is given by
\begin{equation}
    c_{i}= \left(\prod_{j<i}Z_{j} \right)\frac{X_{i}+ i Y_{i}}{2}, \ \
    c_{i}^{\dagger}= \left(\prod_{j<i}Z_{j}\right)\frac{X_{i}- i Y_{i}}{2}, \ \
    c_{i}^{\dagger}c_{i}= \frac{1 - Z_{i}}{2}.
\end{equation}
Using this mapping, Eq.~\eqref{eq:fermi-hubbard} can be rewritten as
\begin{equation}
    H = -\frac{t}{2} \sum_{(i, j)\in E}\sum_{\sigma\in \{\uparrow, \downarrow\}}
    \left(X_{i, \sigma}\vec{Z}X_{j, \sigma}+ Y_{i, \sigma}\vec{Z}Y_{j, \sigma}\right)
    + U\sum_{i \in V}\frac{1-Z_{i, \uparrow}}{2}\frac{1-Z_{i, \downarrow}}{2}
    - \mu\sum_{i\in V}\sum_{\sigma\in\{\uparrow, \downarrow\}}\frac{1-Z_{i, \sigma}}{2},
\end{equation}
where, in $\vec{Z}= \prod_{k}Z_{k}$, the index $k$ runs over the qubits corresponding to the modes between $(i,\sigma)$ and $(j,\sigma)$ in the chosen Jordan--Wigner ordering.

In our work, we have considered the Fermi--Hubbard model on a square lattice with cylindrical boundary conditions. The hopping term is taken as the unit of energy $t=1$, the interaction term is taken as $U/t = 4$, and the chemical potential is taken as $\mu = 0$.

\subsubsection{Random \texorpdfstring{$k$}{k}-local Hamiltonian}
To benchmark the scaling behavior of our compilation costs on unstructured instances, we also consider a synthetic ensemble of random $k$-local Hamiltonians on $N$ qubits.
We generate a Hamiltonian as a sum of $M$ distinct, nontrivial Pauli strings with $\pm 1$ coefficients:
\begin{equation}
    H = \sum_{m=1}^{M} \alpha_m P_m, \qquad \alpha_m \in \{+1,-1\},
\end{equation}
where each $P_m$ is an $N$-qubit Pauli operator of the form
\begin{equation}
    P_m = \bigotimes_{i=1}^{N} \sigma_i^{(m)}, \qquad \sigma_i^{(m)} \in \{I,X,Y,Z\},
\end{equation}
with the locality constraint
\begin{equation}
    \mathrm{wt}(P_m) \coloneqq \left|\left\{ i \in [N] \,:\, \sigma_i^{(m)} \neq I \right\}\right| = k,
\end{equation}
and we explicitly exclude the identity term (i.e., $\mathrm{wt}(P_m)\ge 1$).

In the data provided in GitHub~\cite{github}, we have used $k=4$ and chosen the number of terms as $M=\lfloor N^{3/2}\rfloor$ (e.g., $N=100$ yields $M=1000$), with a fixed random seed for reproducibility.

\subsubsection{Schwinger model (\texorpdfstring{$1+1$}{1+1}D \texorpdfstring{$Z_2$}{Z2} lattice gauge theory)}
A commonly used qubit-representation Hamiltonian~\cite{hondaNegativeStringTension2022} of $1+1$D $Z_2$ lattice gauge theory model is given by
\begin{equation}
    H = J\sum_{n=0}^{N-2}\left[ \sum_{i=0}^{n}\frac{Z_{i}+ (-1)^{i}}{2}+ \frac{\theta_{n}}{2\pi}\right]^{2}
    + \frac{w}{2}\sum_{n=0}^{N-2}[X_{n}X_{n+1}+ Y_{n}Y_{n+1}]
    + \frac{m}{2}\sum_{n=0}^{N-1}(-1)^{n}Z_{n}.
\end{equation}
Here, $J,w,m \in{\mathbb{R}}$ are model parameters, and $\theta_{n}$ is a site-dependent (topological) term.
As mentioned above, the {\tt SELECT} operation does not depend on the coefficient magnitudes, and hence it is uniquely determined (except in the trivial case where a parameter is exactly zero and the corresponding term is absent).
The bottleneck in this Hamiltonian stems from the all-to-all connectivity induced by the first term.
In our data~\cite{github}, we have taken parameters as $J=1, w=1, m = 0, \theta_n = 0.5$.

\subsubsection{2+1D \texorpdfstring{$Z_2$}{Z2} lattice gauge theory}
In addition to the Schwinger model, we consider the (2+1)-dimensional  $Z_{2}$ gauge theory. Following Ref.~\cite{PRXQuantum.3.020320}, we assume that qubits reside on the edges of a square lattice.
For a lattice of size $L\times L$, this yields a total of $2L^{2}$ qubits. The Hamiltonian is defined as the sum of two contributions, corresponding to the electric- and magnetic-field terms,
\begin{align}
    H     & = H_{E}+ h H_{B},                          \\
    H_{E} & = \sum_{i=1}^{2L^2}\bigl( 1- X_{i}\bigr),  \\
    H_{B} & = - \sum_{p} Z_{p_1}Z_{p_2}Z_{p_3}Z_{p_4},
\end{align}
where $i$ runs over lattice edges, $p$ runs over plaquettes, and $(p_1,p_2,p_3,p_4)$ denotes the four edges surrounding a plaquette $p$.

\section{Conversion of Quantum Circuits to Instruction Sequences}
\label{app:conversion}

As mentioned in \cref{subsec:sa_setting}, we convert quantum circuits into instruction sequences consisting solely of $\texttt{CX}$, $\texttt{MAGIC\_MOVE}$, and $\texttt{MAGIC\_MZZ}$. These operations are defined as follows: $\texttt{CX}(c,t)$ is a controlled-$X$ gate acting on control $c$ and target $t$; $\texttt{MAGIC\_MOVE}(t)$ fetches a magic state prepared at an MSF patch to qubit $t$; and $\texttt{MAGIC\_MZZ}(t)$ performs a $ZZ$ measurement between qubit $t$ and a magic state at an MSF patch. In what follows, we describe the conversion rules.

Firstly, for simplicity and clarity, we restrict our attention to compilations that use only \emph{paths}, i.e., pairwise connections realized as lattice-surgery paths. Since multiqubit measurements can be decomposed into sequences of two-qubit operations, this forms a canonical and practically relevant subclass of lattice-surgery compilations. Alternatively, multiqubit gates can be viewed as Steiner trees that connect the terminals corresponding to the involved logical qubits; scheduling and packing of these connections can then be treated as a Steiner tree packing problem~\cite{silvaMultiqubitLatticeSurgery2024}.

Universal quantum computation necessitates logical instructions such as $H$, $S$, $T$, and single-qubit initialization and measurement in the computational basis. A typical implementation of $T$-gates combines magic state preparation with lattice-surgery-based teleportation. Since the remaining primitives occupy localized space and particularly single-qubit initialization and measurement complete within a negligible time compared to lattice-surgery operations, their overhead is significantly lower than that of lattice-surgery operations in large-scale quantum computation~\cite{litinskiGameSurfaceCodes2019}. Consequently, we focus our analysis on magic state throughput and lattice-surgery scheduling, omitting the contributions of localized single-qubit primitives.

Based on the above, the multiqubit unitary gates appearing in the input circuit can be replaced as follows. First, controlled-$Y$ and controlled-$Z$ gates are reduced to controlled-$X$ via conjugation with $S$ and $H$ gates, respectively. Since these single-qubit gates are neglected in our setting, we can treat all such controlled operations as $\texttt{CX}$.
Furthermore, the three-qubit Toffoli gate (CCX), which acts on two control qubits $c_{1},c_{2}$ and a target qubit $t$, can similarly be replaced using $\texttt{MAGIC\_MOVE}$ and $\texttt{CX}$ gates, as shown in \citet[Fig.~S20(a)]{yoshiokaHuntingQuantumclassicalCrossover2024}. The remaining $T$ gate can also be replaced with $\texttt{MAGIC\_MZZ}$.
Although $S$-gate feedback is neglected in this configuration, a feedback challenge arises when the determination of the measurement outcome is significantly delayed, as discussed in \cref{subsec:space_time_routing_limitation}.

By applying the above conversion rules, the quantum circuit is transformed to contain only $\texttt{CX}$, $\texttt{MAGIC\_MOVE}$, and $\texttt{MAGIC\_MZZ}$ operations.

\section{Supplement for Architectures}
\label{app:other_architectures}

In this section, we present the detailed information on the architecture.

\subsection{Floorplan}\label{app:floorplan}

\begin{figure}[t]
    \centering
    \begin{minipage}{0.3\columnwidth}
        \centering
        \textbf{$1/4$-floorplan} \includegraphics[height=\columnwidth]{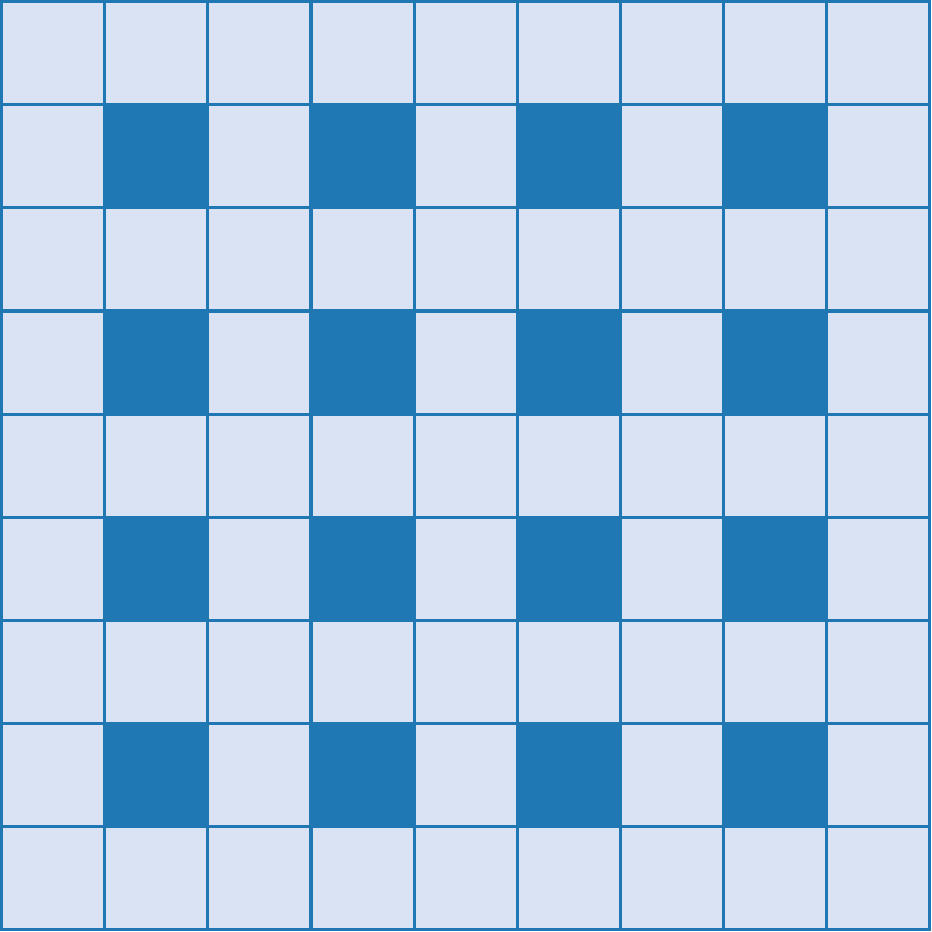}
    \end{minipage}%
    \hspace{0.04\columnwidth}%
    \begin{minipage}{0.3\columnwidth}
        \centering
        $4/9$-floorplan \includegraphics[height=\columnwidth]{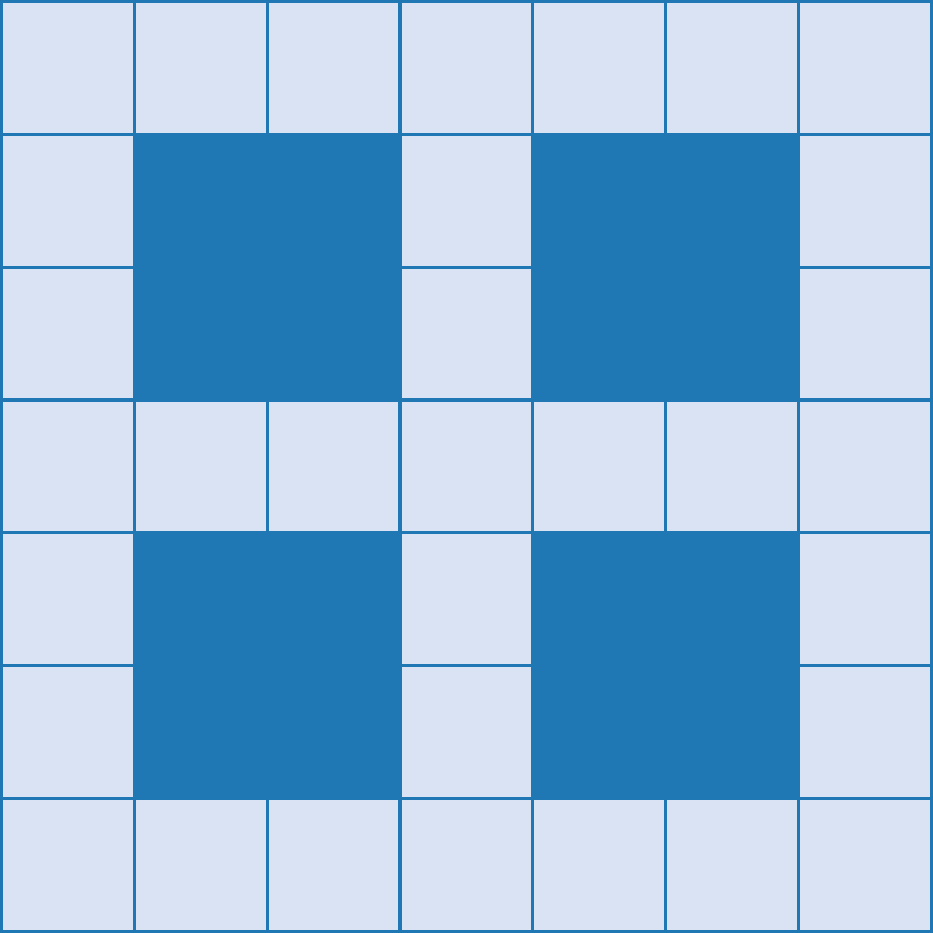}
    \end{minipage}%
    \hspace{0.04\columnwidth}%
    \begin{minipage}{0.3\columnwidth}
        \centering
        $1/2$-floorplan \includegraphics[height=\columnwidth]{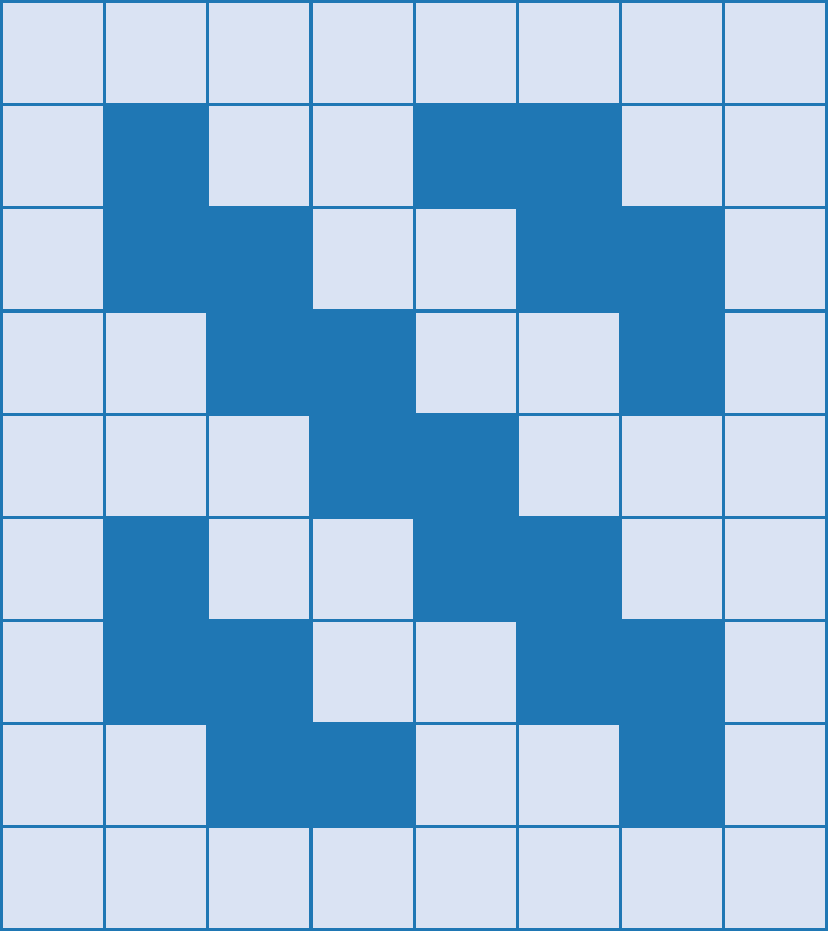}
    \end{minipage}
    \caption{The floorplans for the surface-code-based architectures. The dark and light patches represent the data and bus patches, respectively. This paper adopts the $1/4$-floorplan.}
    \label{fig:floorplan}
\end{figure}

The underlying layout of data and bus qubits in lattice surgery is termed a floorplan, reflecting the inherent flexibility in how these components are arranged on the surface code.
Several floorplans allow connections to both $X$-boundaries and $Z$-boundaries without rotation, such as the $1/4$-floorplan \cite{beverlandSurfaceCodeCompilation2022}, the $4/9$-floorplan \cite{chamberlandUniversalQuantumComputing2022}, and the $1/2$-floorplan \cite[Fig. 2]{uenoHighPerformanceScalableFaultTolerant2024} (see \cref{fig:floorplan}).
Since this work considers general quantum circuits as inputs (requiring the use of both $X$- and $Z$-boundaries) and focuses on the optimization of scheduling and mapping, we adopt the $1/4$-floorplan.
It is worth noting that there exists a trade-off between qubit occupancy and path collision rate: the lower the fraction of occupied sites assigned to data qubits, the greater the room for optimization. For this reason, we focus on layouts with relatively low data-qubit occupancy, thereby treating the problem as a classical optimization problem with substantial flexibility. In contrast, when the occupancy is high, optimal routing can often be derived almost trivially from the input.

In our numerical experiments, the data chip is modeled as a rectangular lattice of size $w\times h\times l$, where the number of layers satisfies $l\in\{1,2\}$. The case $l=1$ corresponds to the 2D layout, whereas $l=2$ represents the 2.5D layout consisting of two layers. Recall that $n$ is the number of logical data qubits (\cref{subsec:sa_setting}).
The values of $w$ and $h$ are defined by
\begin{equation*}
    w=h=\qty(2\left\lceil\sqrt{n/l}\right\rceil+1)+2.
\end{equation*}
The constant $+2$ arises because the perimeter is occupied by MSF as in \cref{fig:layout,tab:mapping_images}.
These $w,h$ guarantee that the layout provides at least
\begin{equation*}
    l \frac{w-3}{2}\frac{h-3}{2} = l \left\lceil\sqrt{n/l}\right\rceil^2 \ge n
\end{equation*}
patches for data qubits.
Recall that $X_{\mathrm{MSF}}$ denotes the set of patches on which MSFs are actually placed, and $n_{\mathrm{MSF}}= \abs{X_{\mathrm{MSF}}}$ is the number of MSFs used. Specifically, we impose the constraint
\begin{equation}
    n_{\mathrm{MSF}} = 4\left\lceil\sqrt{n}\right\rceil+4
\end{equation}
for all cases. This value corresponds to the maximum number of MSFs available when $l=1$, and we therefore restrict the number of MSFs accordingly when $l=2$.

\subsection{2.5D layout}
\label{sec:25d_layout}

In contrast to conventional 2D qubit plane architectures, recent theoretical and experimental proposals have introduced more complex architectures stacking a couple of 2D qubit layers~\cite{pattisonHierarchicalMemoriesSimulating2025,rametteFaulttolerantConnectionErrorcorrected2024,haArchitecturesLatticeSurgerybased2025,viszlaiInterleavedLogicalQubits2025,uenoHighPerformanceScalableFaultTolerant2024,liao2026breakingscalabilitybarriervertical}, which we call a 2.5D layout.
Among them, we focus in particular on a two-layer configuration in this paper.

To fully exploit the enhanced qubit connectivity of the 2.5D layout, it is vital to incorporate inter-layer operations into the routing framework. 
An example of such inter-layer operations is the transversal CNOT gate~\cite{pattisonHierarchicalMemoriesSimulating2025,rametteFaulttolerantConnectionErrorcorrected2024,viszlaiInterleavedLogicalQubits2025}.
It can be performed between two adjacent patches across layers in either direction, provided their boundary orientations are identical. 
Since transversal CNOTs are completed in a single code cycle, they are treated as negligible when quantifying lattice-surgery costs in code beats.
Note that a 2.5D architecture supporting transversal CNOTs requires $O(d^2)$ inter-layer connections per surface code patch.

As an alternative inter-layer operation, we can employ inter-layer lattice surgery~\cite{haArchitecturesLatticeSurgerybased2025,uenoHighPerformanceScalableFaultTolerant2024,viszlaiInterleavedLogicalQubits2025}.
Whereas conventional 2D lattice surgery merges spatially adjacent patches, the 2.5D configuration enables the merging of ``diagonally'' adjacent patches across different layers; specifically, this layout allows for simultaneous movement in both spatial and inter-layer directions. 
A 2.5D architecture supporting this technique requires the two layers to be connected only along the edges of the surface code patches. 
Consequently, the required inter-layer connections per patch are reduced to $O(d)$, offering a distinct hardware cost advantage compared to the transversal CNOT approach.

A notable related work is the \emph{Bypass architecture} proposed by Ref.~\cite{uenoHighPerformanceScalableFaultTolerant2024}, which leverages a biplanar architecture consisting of a dense Logic layer for main computation and a sparse Bypass layer for lattice-surgery routing. 
Compared with conventional stacking of qubit planes, their design reduces effective path lengths via the sparsity of the Bypass layer, thereby improving fidelity and execution time. 
Since their evaluation employs a greedy scheduling policy \cite[Sec.~3.1]{uenoHighPerformanceScalableFaultTolerant2024} and random logical-qubit allocation \cite[Sec.~7.2]{uenoHighPerformanceScalableFaultTolerant2024}, applying our optimized mapping and routing techniques to this architecture represents an intriguing prospect for future work to achieve even higher performance.

\section{Numerical experiment on impact of magic state preparation time}

In this section, we numerically compare the architectures and outline the general trend of their performances.

\subsection{Comparison of the \texorpdfstring{$\tau$}{tau} value}
\label{subsubsec:tau_comparison}

This section investigates the impact of the magic state preparation time $\tau$ introduced in \cref{sec:performance_metric}.
This parameter is particularly critical in the early FTQC era, where non-negligible error rates will likely necessitate multiple attempts to generate a magic state. To model this situation, a larger value of $\tau$ must be considered. We therefore compared the architectures' performance with $\tau$ ranging from $0$ to $10$ in increments of two, reflecting scenarios from ideal to error-prone conditions.

We use the following settings to compare the architectures for different $\tau$ values.
\begin{description}
    \item[Circuit] SELECT-6 circuit for 2D Heisenberg model.

    \item[Allocation] Simulated annealing.

    \item[Routing] Double-slice routing (correcting kink parity).
\end{description}
We adopted these settings to obtain optimized results. Specifically, we chose double-slice routing for comparison because it yields a solution with a small circuit volume, making it particularly well-suited for the early FTQC era, where a larger value of $\tau$ should be considered.

\begin{figure}
    \centering
    \includegraphics[width=0.5\linewidth]{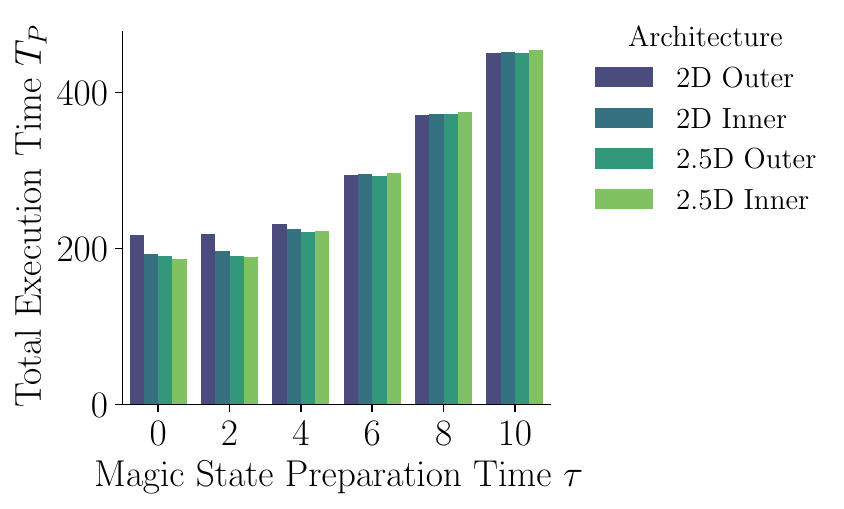}
    \caption{Influence of magic state preparation time $\tau$ on the relative performance of the four architectures. The plot compares the total execution time for the SELECT-6 circuit for the 2D Heisenberg model to assess how the performance gap between architectures changes as $\tau$ is increased. All results were generated using double-slice routing and SA allocation.
    }
    \label{fig:prep_time}
\end{figure}

\Cref{fig:prep_time} shows the effect of magic state preparation time on the performance comparison of various architectures. The execution time increases with the magic state preparation time, as the number of MSFs is fixed. This growing bottleneck reduces the performance differences between the architectures, thus preventing a meaningful comparison.

It is anticipated that increasing the magic state preparation time and the number of MSFs while maintaining the overall generation rate would result in a smooth transition from the ideal inner factory layout to the outer factory layout. Specifically, when the generation rate is slow and the number of MSFs is large, internal MSF embedding loses its meaning. In such cases, the system would function by bundling the outer MSFs and fetching only the successful ones, which in turn mimics the behavior of the outer factory layout.

\subsection{General Performance Trends}

This section presents the general performance trends of the architectures. As expected from their designs, 2.5D layouts outperformed 2D layouts, and inner factory layouts outperformed outer factory layouts. The superiority of 2.5D layouts was more conspicuous than that of inner factory layouts.
\begin{figure}
    \centering
    \includegraphics[width=0.4\linewidth]{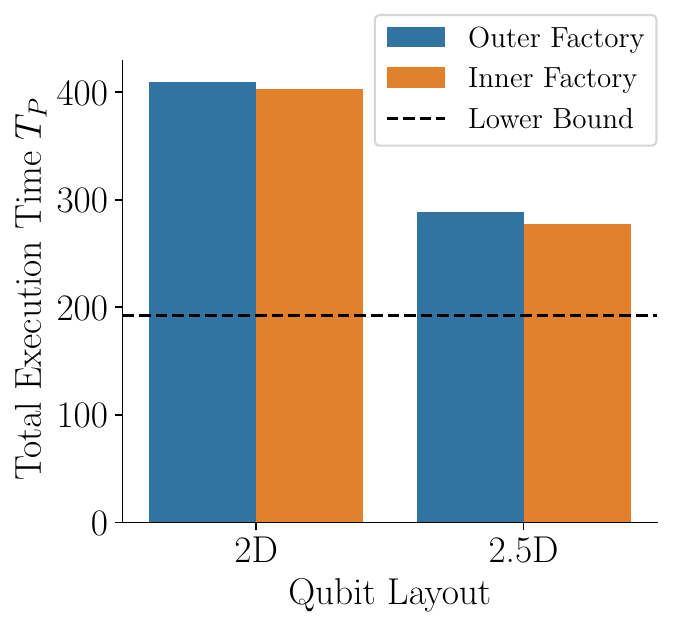}
    \caption{Comparison of layouts when SA-based mapping and double-slice routing are applied to the SELECT-6 circuit for random Hamiltonian. The magic state preparation time $\tau$ is set to $2$ in this experiment.}
    \label{fig:select_random_local_archi_comparison}
\end{figure}

The performance of 2.5D layouts shows a clear advantage when dealing with complex routing demands, as illustrated by the SELECT circuit for random Hamiltonian (\cref{fig:select_random_local_archi_comparison}). This is a natural consequence of its ability to mitigate routing collisions. A similar trend is observed in the SELECT circuit for the Schwinger Hamiltonian.

\begin{figure}
    \centering
    \includegraphics[width=0.5\linewidth]{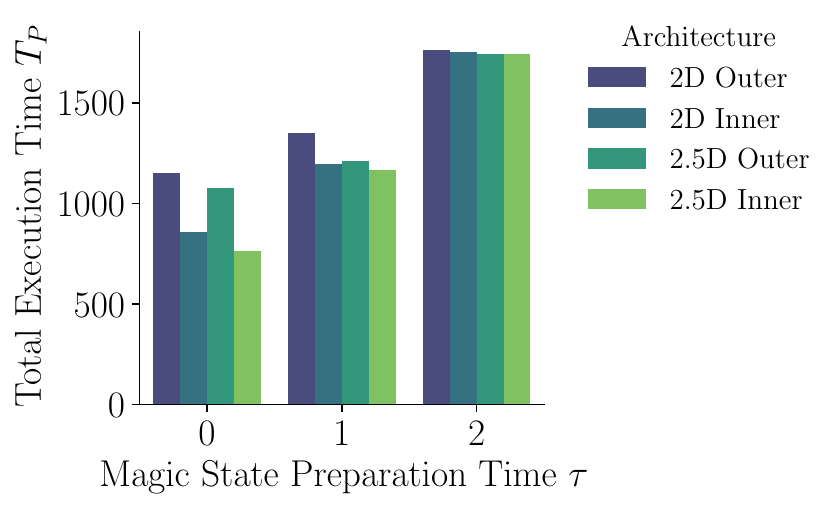}
    \caption{Performance comparison of the four architectures on the Trotter circuit for the 2D Heisenberg model based on the deterministic decomposition. The plot shows total execution time as a function of magic state preparation time. The four architectures are distinguished by color. All results were generated using double-slice routing and SA allocation.}
    \label{fig:trotter_msf_prep_0_to_2}
\end{figure}

Conversely, the performance advantage of inner factory layouts appears less pronounced compared to that of 2.5D layouts. A noticeable improvement is observed for several circuits, such as the SELECT circuit for the Heisenberg Hamiltonian, as shown in \cref{fig:prep_time}.

It is noteworthy that, although it may not be a realistic scenario, the Trotter circuit for the Heisenberg Hamiltonian with $\tau = 0$ exhibits a significant advantage of inner factory layouts (\cref{fig:trotter_msf_prep_0_to_2}). This finding suggests that, at a larger circuit scale and with a sufficient number of MSFs, the Trotter circuit could also demonstrate the benefits of inner factory layouts under realistic conditions.

\section{Routing Details}
\label{app:routing}

\subsection{Extension of Spacetime Routing to 2.5D layouts}
\label{subsec:35D_routing}

The spacetime routing technique hinges on the kink condition, a structural constraint on spacetime paths. While this condition arises as a necessary requirement for consistently assigning measurement bases to each spatial path segment, it is also sufficient to realize intended two-body operations, such as logical qubit movement, joint $XX$ and $ZZ$ measurements, and CNOT gates. 
In the remainder of this section, we first introduce the underlying circuit equivalence in \cref{subsubsec:circuit_equivalence}, followed by the reduction for 2D layouts in \cref{app:spacetime_2d_simplification} and its generalization to 2.5D layouts in \cref{subsubsec:spacetime_routing_extension_proof}. 

\subsubsection{Simplification of Measurement Chains}
\label{subsubsec:circuit_equivalence}

To prove sufficiency, we utilize the equivalence between a chain of two-body measurements and a single two-body operation. This equivalence is implicit in the simplification procedure of Ref.~\cite[Section 3.3.2]{hamadaEfficientHighperformanceRouting2024}, and we formalize the result as follows.
\begin{lemma}[{\cite[Section 3.3.2]{hamadaEfficientHighperformanceRouting2024}}]
    \label{lem:simplify_chain}
    Consider a sequence of qubits $Q_0, \dots, Q_N$ linked by two-body Pauli measurements $M_1, \dots, M_N$, where each $M_i$ is a joint $XX$ or $ZZ$ measurement on $Q_{i-1}$ and $Q_i$. The execution order of these measurements is arbitrary. Qubits $Q_0$ and $Q_N$ serve as targets, while the intermediate qubits $Q_1, \dots, Q_{N-1}$ are ancillae initialized in $\ket{0}$ or $\ket{+}$ and measured in Pauli $X$ or $Z$ basis. If the initialization (resp.\ measurement) basis of each ancilla is distinct from the basis of the succeeding (resp.\ preceding) two-body measurement, the circuit implements an effective two-body operation on $Q_0$ and $Q_N$ via appropriate Pauli feedback. The resulting operation depends on the boundary types of $Q_0$ and $Q_N$, as shown in \cref{tab:simplify_chain}.
\end{lemma}

\begin{table}[tb]
    \centering
    \caption{Summary of effective two-body operations implemented by a chain of two-body measurements as described in \cref{lem:simplify_chain}.}
    \label{tab:simplify_chain}
    \begin{tabular}{ccc}
        \toprule
        \; Boundary Type of $Q_0$ \; & \; Boundary Type of $Q_N$ \; & Effective Operation \\
        \midrule
         $X$  &  $X$  & $XX$ \\
         $X$  &  $Z$  & CNOT (Control: $Q_N$, Target: $Q_0$) \\
         $Z$  &  $X$  & CNOT (Control: $Q_0$, Target: $Q_N$) \\
         $Z$  &  $Z$  & $ZZ$ \\
        \bottomrule
    \end{tabular}
\end{table}

Note that logical qubit movement can be reduced to either an $XX$ or $ZZ$ measurement, depending on the boundary type, since both measurement types can transport a logical data qubit from one end to the other under appropriate pre- and post-processing. Thus, we only need to verify cases where the target operation is an $XX$ or $ZZ$ measurement, or a CNOT gate.

While this lemma is sufficient for verifying the spacetime routing technique in 2D architectures, we slightly generalize it to accommodate 2.5D architectures equipped with transversal CNOTs.

\begin{remark}
    \label{rem:cnot_chain_simplification}
    \cref{lem:simplify_chain} generalizes to cases where each $M_i$ may be a CNOT gate. This extension is justified by decomposing each CNOT gate into its equivalent pair of $XX$ and $ZZ$ measurements. Consequently, we treat the basis of CNOT controls and targets as $Z$ and $X$, respectively, and determine the initialization and measurement basis of ancillae as in \cref{lem:simplify_chain}.
\end{remark}

\subsubsection{Reduction for 2D layouts}
\label{app:spacetime_2d_simplification}

We begin by focusing on 2D layouts. Given a 3D path that satisfies the three conditions presented in \cref{subsec:overview_3drouting}, we verify that such a path realizes the intended two-body operations, specifically joint $XX$ and $ZZ$ measurements and CNOT gates.

A spacetime path can be decomposed into an alternating sequence of spatial segments and temporal segments, which are joined at their endpoint voxels. Each spatial segment is assigned an $XX$ or $ZZ$ measurement, while each temporal segment corresponds to a bus qubit acting as an ancillary logical qubit.

As discussed in the main text, the measurement bases of the spatial segments are determined sequentially from the path ends. Recalling that a kink is defined as a temporal segment connecting two different boundary types at its ends, a temporal segment switches the measurement type if and only if it is a kink. Consequently, the boundary condition determines the measurement types of the two endpoint spatial segments, and the kink condition ensures that these measurement bases can be consistently assigned.

This assignment yields a chain of joint $XX$ and $ZZ$ measurements, which can be simplified to a single two-body operation by following the initialization and measurement protocol described in \cref{lem:simplify_chain}. Moreover, the boundary conditions ensure that the resulting operation is the desired one under appropriate Pauli feedback. 

Finally, we demonstrate that the resources consumed by the operation correspond to the voxels occupied by the spacetime path. To this end, we analyze the computational cost of pre- and post-processing. Since the initialization and measurement of ancillae are completed in a single code cycle, the cost is negligible compared to the $d$ code cycles required for a single lattice-surgery operation. Furthermore, Pauli feedback can be implemented via the Pauli frame~\cite{FowlerGidney2018LowOverhead,Fowler2012TowardsPractical}, which shifts the computational burden solely to classical processing.

\subsubsection{Extension to 2.5D layouts}
\label{subsubsec:spacetime_routing_extension_proof}

We now extend this spacetime routing technique to 2.5D layouts. Regarding the inter-layer operations of the underlying 2.5D architecture, we consider inter-layer lattice surgery and transversal CNOT operations, specifically employing the latter in our implementation. For convenience, we restate the theorem below. Notably, this statement naturally generalizes to multiplanar configurations, provided that adjacent layers maintain the same connectivity as in the biplanar setting.

\spacetime*

\begin{proof}
    We will prove the statement for each of the two cases.

    \noindent
    \textbf{Proof of Case (i):}
    The proof for the 2D case extends naturally to this scenario. A spacetime path similarly decomposes into an alternating sequence of spatial and temporal segments, where each spatial segment is a 2.5D path that may traverse both layers and each temporal segment corresponds to a bus qubit. Then, following the same argument as in the 2D case, we assign measurement types to the spatial segments and reduce the resulting chain to the desired operation.

    \noindent
    \textbf{Proof of Case (ii):}
    In this scenario, a spacetime path decomposes into an alternating sequence of spatial and temporal segments, where each spatial segment is constrained to a single layer of the 2.5D layout. Unlike the 2D case, where a temporal segment lies on a single bus qubit in the 2D case, it spans two bus qubits adjacent in the inter-layer direction when transversal CNOTs are included. 

    As a crucial requirement for realizing a transversal CNOT gate, its two operand patches must share the same boundary orientation. Consequently, the measurement types of two spatial segments linked by a transversal CNOT differ if and only if the CNOT forms a kink. The kink condition thus remains a necessary condition for consistently assigning measurement types to spatial segments.
    
    After assigning these measurement types, we address the inter-layer transitions to convert the path into a circuit. To this end, two voxels adjacent in the inter-layer direction are treated as a single transversal CNOT within these spacetime resources. Given the initialization and measurement procedures in \cref{lem:simplify_chain}, the CNOT gate direction is arbitrary. Furthermore, because the execution time of a transversal CNOT is negligible relative to a single code beat, the operation can be performed at any point within the corresponding one code beat. This flexibility allows for elimination of one voxel, as detailed in \cref{rem:adjacency_relaxation}.
    The resulting chain of joint $XX$ and $ZZ$ measurements and CNOT gates simplifies to the desired two-body operation via \cref{rem:cnot_chain_simplification} and the boundary conditions. 
\end{proof}

Notably, while the path condition for 2D layouts requires adjacent voxels in a path to share a common face, this condition for 2.5D layouts can be slightly relaxed (see \cref{fig:omissible_voxel}). \cref{rem:adjacency_relaxation} states a sufficient condition for omitting a voxel from a spacetime path, although this optimization is not utilized in our implementation.

\begin{remark}
    \label{rem:adjacency_relaxation}
    The path condition for 2.5D layouts may be relaxed by omitting a voxel $v$ from an inter-layer transition, provided that specific connectivity requirements are met. Let $u$ be the paired voxel in the transition and $w$ the neighbor of $v$ distinct from $u$. 
    \begin{itemize}
        \item \textbf{Case (i):} $v$ can be omitted if $v$ and $w$ are spatially adjacent. This can be easily verified via the architectural connectivity between $u$ and $w$ (see \cref{sec:25d_layout}). 
        
        \item \textbf{Case (ii):} $v$ can be omitted if $v$ and $w$ are temporally adjacent. Specifically, if $w$ precedes (resp. succeeds) $v$, a transversal CNOT must be applied at the beginning (resp.\ end) of $u$ and $v$. 
    \end{itemize}
    After the elimination of $v$, the remaining voxels $u$ and $w$ share only an edge, rather than a face.
\end{remark}

\begin{figure}
    \centering
    \includegraphics[width=\linewidth]{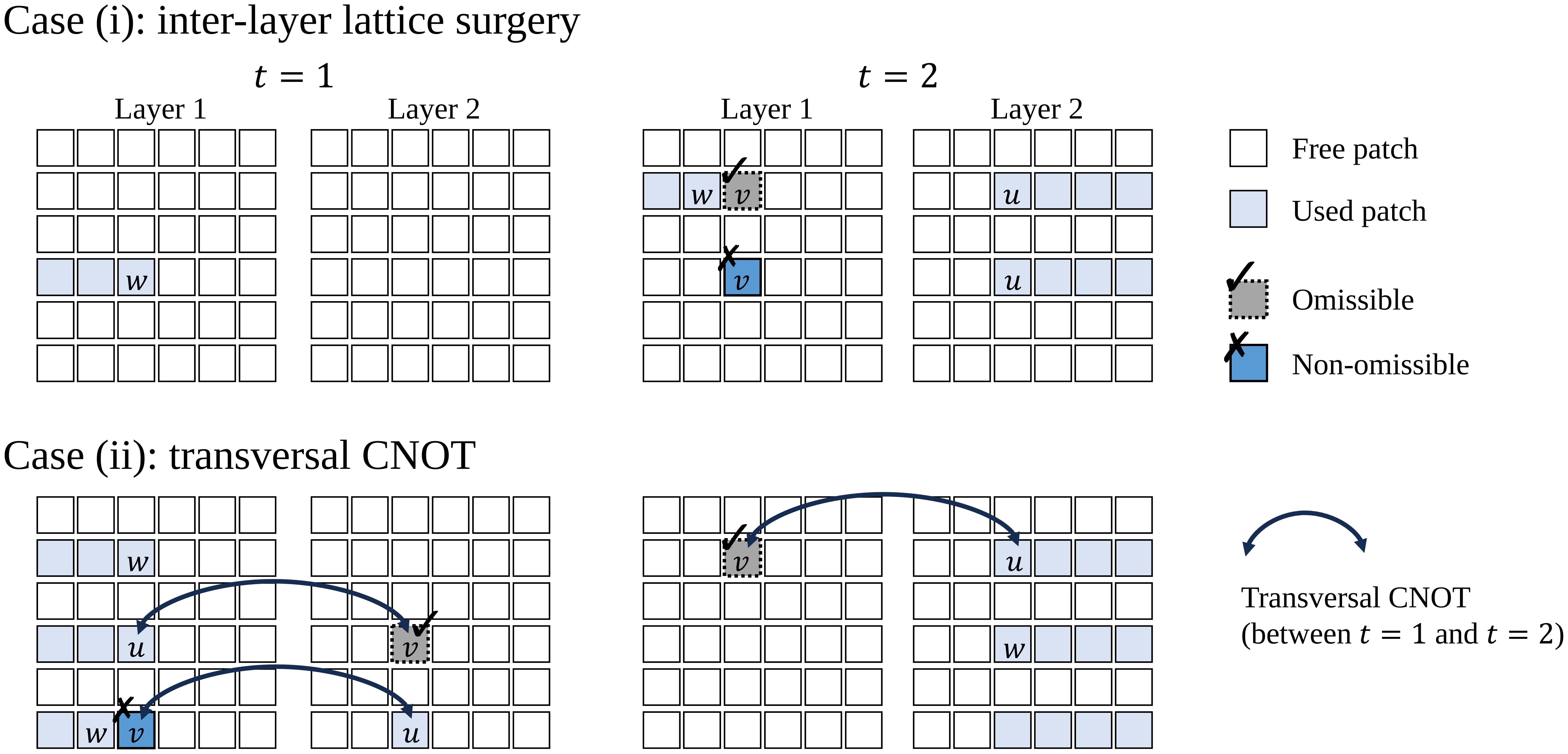}
    \caption{Example of omissible voxels in 2.5D layouts.}
    \label{fig:omissible_voxel}
\end{figure}

\subsection{Termination Guarantee of Kink Parity Correction for Double-slice Routing}
\label{subsec:termination_guarantee}

In this section, we demonstrate that the kink parity correction of double-slice routing succeeds under three assumptions: (i) the two target qubits have the same boundary orientation, (ii) none of the temporally adjacent voxels in the two focused time slices is blocked by other paths, and (iii) the path before correction is the shortest path. As a corollary, provided that all data qubits share a common boundary orientation, the entire algorithm is guaranteed to terminate successfully. This is because routing for each instruction will eventually be attempted in two time slices free from obstruction if the algorithm repeatedly fails to route the path or correct the kink parity.

Since the argument for 2.5D architectures equipped with inter-layer lattice surgery is similar to that for 2D architectures, we assume that 2.5D architectures support transversal CNOTs as inter-layer operations in the following proof. While the following discussion focuses on the biplanar configuration, its extension to multiplanar architectures can be achieved by using a similar technique as in \cref{subsubsec:proj_kink_25d}.

Our proof relies on the following claims:
\begin{enumerate}[leftmargin=*, label=\textit{Claim~\arabic*}.]
    \item Any path connecting identical (resp.\ different) boundary types has an even (resp.\ odd) number of 90-degree corners.
    \label{claim:degree}

    \item The path lies on a single time slice and contains at most one inter-layer movement. Consequently, the path has at most one kink due to inter-layer movement.

    \item The kink parity of the path can be toggled if the first 90-degree corner encountered from an endpoint is not a kink. Specifically, pinching the two voxels forming the corner closer to the endpoint increases the number of kinks by one.
\end{enumerate}
Claim 1 directly follows from assumption (i), whereas Claim 2 follows from assumptions (ii) and (iii). The adjustment in Claim 3 is always feasible owing to assumption~(ii). While one might be concerned that this adjustment could lead to self-intersection, such a case contradicts assumption~(iii), which imposes that the initial path is a shortest path.

Based on these facts, we establish the proof for each operation in 2D and 2.5D layouts. We first prove that the kink parity can be successfully corrected for $XX$ and $ZZ$ measurements, to which logical qubit movement is reduced as well. In this case, the kink parity should be corrected to even. In 2D layouts, Claim 2 ensures that the path contains no kinks, satisfying the kink condition. In 2.5D layouts, we show that the kink parity can be corrected when the path has one kink. As Claim 1 guarantees that the path has an even number of 90-degree corners, we have at least one non-kink corner. By traversing the path from at least one of its endpoints, a non-kink corner can be reached first. Then, Claim 3 guarantees that the kink parity can be corrected.

We then prove that the kink parity can be corrected for CNOT operations, where the kink parity should be corrected to odd. In both 2D and 2.5D layouts, Claim 2 ensures that the correction is required only when the path has no kinks. Since Claim 1 guarantees that the path has an odd number of 90-degree corners, we can apply Claim 3 to toggle the kink parity. This completes the proof of successful correction.

Finally, we outline the extension to the multiplanar setting. \cref{subsubsec:proj_kink_25d} defines \textit{adjustable} kinks as those that do not involve inter-layer transitions. Given the assumption on the shared boundary orientation, any path without non-kink corners must satisfy the kink condition. Hence, if a path does not satisfy the kink condition, it must contain a non-kink corner. Consequently, analogous to Claim 3, pinching the first corner excluding non-adjustable kinks successfully corrects the kink parity, because all preceding corners are non-adjustable kinks.

\subsection{Extension of Routing Methods to 2.5D Layouts}

Here, we discuss the extension of the three routing methods described in \cref{subsec:routing_methods}.
While the routing strategies can be straightforwardly extended to 2.5D layouts by changing the adjacency of patches, it is nontrivial to extend the kink parity correction techniques. Note that the following extension naturally accommodates multiplanar architectures and is not restricted to biplanar configurations.

\subsubsection{Kink Parity Correction for Single-slice Routing}
\label{subsubsec:single_kink_25d}

While the long-range CNOT approach is also applicable to 2.5D architectures, we can develop a more efficient yet heuristic correction technique for 2.5D architectures accepting transversal CNOTs, because kinks can be formed by crossing qubit plane layers. 

To exploit this capability, we implement a strategy to generate a set of candidate paths and select the shortest one that satisfies the kink condition. These candidate paths are generated by specifying two adjacent patches and finding the shortest path traversing them. While this approach is efficient and generally effective, its success depends on the ancillary routing space.

Our implementation of single-slice routing for 2.5D layouts also utilizes this technique, since it assumes that the 2.5D layouts are equipped with transversal CNOTs. In our evaluation, the method worked flawlessly with negligible overhead for the 1/4-floorplan, which is our primary focus. However, we observed instances where the approach failed to find a valid path in a 4/9-floorplan, where the ancillary routing space is more constrained.

\subsubsection{Kink Parity Correction for Projective Routing}
\label{subsubsec:proj_kink_25d}

We first discuss the correction technique for 2D layouts, introduced by Ref.~\cite{hamadaEfficientHighperformanceRouting2024}. Because spacetime paths are constructed by stacking a 2D path onto the current scheduling result in the time direction, the voxels above the spacetime path are always available. The correction technique exploits this free space to flip the kink parity through the following steps.
\begin{enumerate}[leftmargin=*, label={Step~\arabic*}.]
    \item Traverse the spacetime path from one end and identify the initial 90-degree corner, which may be a kink.
    \item If this corner is not a kink, pinch the corner voxel and its predecessor in the time direction. As this adjustment always flips the kink parity, terminate the procedure.
    \item If this corner is a kink, align the corner voxel and its two neighbors at their maximum height to ensure the corner is no longer a kink.
    \item If the kink parity remains unchanged after Step 3, apply the adjustment from Step 2 to the resulting non-kink corner, then terminate the procedure.
\end{enumerate}
As the choice of the starting endpoint for traversal in Step 1 is arbitrary, our implementation prioritizes the direction such that the first corner is not a kink. If both possible first corners are of the same type (either both kinks or both non-kink corners), the one with the lower height is selected.

Assuming that the two target qubits share a common boundary orientation, this procedure is guaranteed to flip the kink parity. This can be verified by establishing the following claims:
\begin{enumerate}[leftmargin=*, label=\textit{Claim~\arabic*}.]
    \item If the initial spacetime path does not satisfy the kink condition, it has at least one 90-degree corner.
    \item The adjustment from Step 2 to the first 90-degree corner always flips the parity.
\end{enumerate}
Claim 1 can be verified by considering two cases: if the operand boundaries differ, the existence of a 90-degree corner follows directly; if the boundaries are the same, the kink parity is odd due to the violation of the kink condition, which implies the existence of a 90-degree corner. Claim 2 is justified by the choice of pinched voxels. Although pinching two voxels only affects corners centered at those voxels, selecting the first corner ensures that the predecessor cannot be the center of a corner. This completes the proof of the entire correctness.

We now generalize this correction technique to 2.5D layouts. For architectures that support inter-layer lattice surgery, the same methodology remains valid. In the following discussion, we focus on transversal CNOTs.

The key idea is to distinguish between adjustable and non-adjustable kinks. A kink is termed \textit{adjustable} if its temporal transition does not include inter-layer movement. Notably, non-adjustable kinks persist as kinks regardless of any temporal voxel adjustments.

Consequently, the only procedural modification is restricted to Step 1: the search for the initial 90-degree corner is updated to skip non-adjustable kinks. The proof of correctness follows from an analogous discussion. As a stronger version of Claim 1, we can prove the existence of a non-kink corner as follows: as the kink condition is violated, not all corners are kinks, which in turn ensures the existence of a non-kink corner. Claim 2 remains valid because any preceding corners are non-adjustable kinks. This completes the correctness proof for the generalized correction technique, assuming shared boundary orientations.

\section{CBPI Stack}

\subsection{Computation of CBPI Stack}
\label{sec:cbpi_detail}

In this section, we detail the computation of the base and each hazard. To quantify the effects of the four hazards in terms of CBPI and circuit volume, we execute the scheduler under the following five incremental scenarios, ranging from a theoretical ideal to the final feasible solution. Each hazard is then calculated as the performance degradation between consecutive scenarios.
\begin{enumerate}
    \item \textit{Base}: The first scenario provides the ideal performance where all hazards are absent. We assume that each instruction exclusively demands one code beat (i.e., one voxel) for each of its operands. Operands can be processed asynchronously, MSFs are assumed to be infinite, and spatial routing constraints are omitted. Consequently, the total execution time equals the maximum number of instructions applied to any single data qubit, and the total circuit volume is the product of this execution time and the number of data qubits.

    \item \textit{Operand Synchronization}: The second scenario accounts for the synchronization of operands imposed by single- or double-slice routing, representing the ideal performance of each routing framework. Specifically, single-slice routing requires all operands of each instruction to be completely synchronized, whereas double-slice routing slightly relaxes this constraint, limiting the difference in their execution timings to at most one code beat. At this stage, each instruction still demands data qubit occupation for one code beat, and spatial routing and MSF constraints remain excluded. Crucially, as single-slice routing requires two code beats for a CNOT gate in 2D layouts, this and subsequent scenarios assume that each operation consumes two code beats, effectively doubling the raw scheduler output.

    \item \textit{CX Path Congestion}: The third scenario introduces spatial constraints for all $\texttt{CX}$ instructions within the bus qubit space; specifically, the scheduler finds a routing path between operands via bus qubits. In contrast, the scheduler does not yet attempt to route paths involving an MSF patch. Furthermore, the kink condition and its associated correction techniques are ignored at this stage.
        
    \item \textit{Magic Path Congestion}: The fourth scenario considers the finite throughput of the MSF patches and the spatial congestion due to their incident paths. We now require each operation to occupy $(\tau + 1)$ code beats of an MSF patch if it consumes a magic state. The kink condition and its associated correction techniques are still ignored at this stage.
    
    \item \textit{Kink Parity Correction}: The fifth scenario finally observes all constraints, including the kink condition. This represents the actual execution time and circuit volume of the feasible scheduler output.
\end{enumerate}

\subsection{Breakdown of Performance Comparison between Outer and Inner Layouts}
\label{sec:breakdown_outer_inner}

\Cref{fig:magic_breakdown} presents the performance breakdown of the three routing algorithms under different factory layouts and $\tau$ values, evaluated using 2D layouts and the SELECT-6 circuit simulating the 2D Heisenberg model. While inner layouts successfully mitigate the magic path hazard as expected, their advantage diminishes as $\tau$ increases from $0$ to $2$. Recall that the combination of projective routing and inner layouts may cause severe spatial congestion, as discussed in \cref{subsec:space_time_routing_limitation}, which is neglected in these numerical experiments.

\begin{figure}
    \centering
    \includegraphics[width=\linewidth]{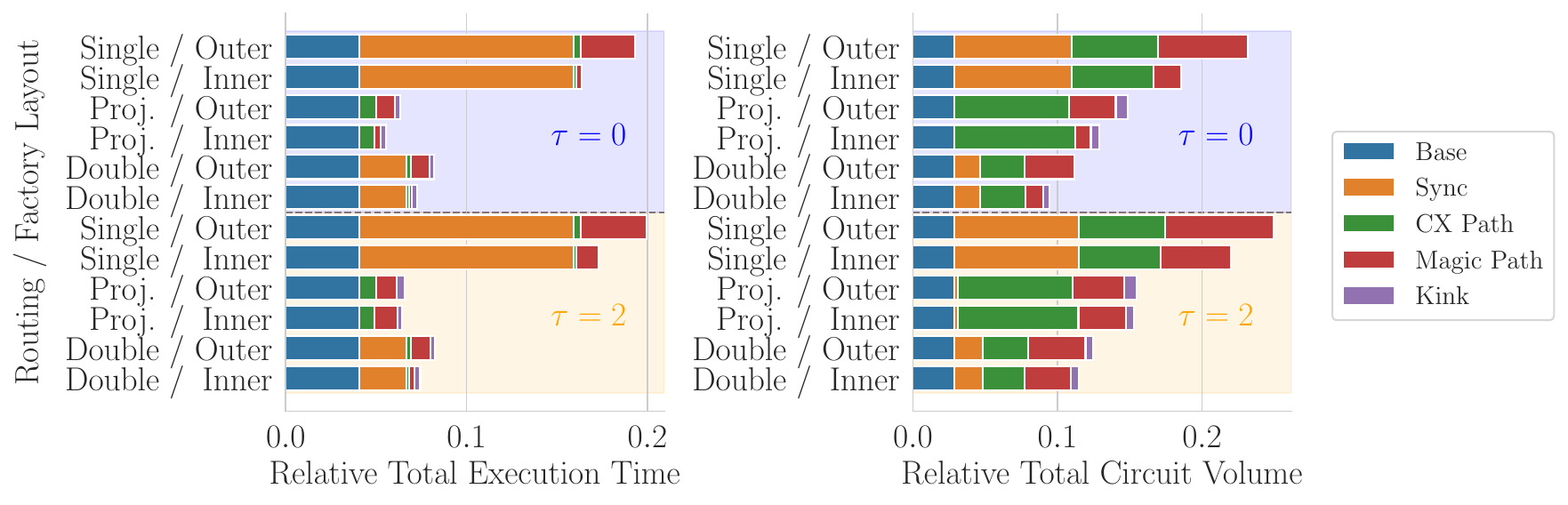}
    \caption{Performance breakdown of the three routing algorithms under different factory layouts and $\tau$ values, evaluated using 2D layouts and the SELECT-6 circuit simulating the 2D Heisenberg model. For enhanced clarity, all values are normalized relative to the EDPC baseline, which is omitted from the plot for visibility.}
    \label{fig:magic_breakdown}
\end{figure}

\section{Mapping}
\label{app:mapping}

In this section, we present the solution procedure and detailed numerical results for the mapping optimization problem \eqref{prob:mapping_optimization_problem} discussed in \cref{sec:mapping}. Recall that $\{x_i\}_{i=1}^n \subseteq X$ specifies the assignment of each qubit $i$ to one of the candidate patches in $X$ on the qubit plane.

\subsection{Solution Method}
\label{app:mapping_how_to_solve}

We solve Problem~\eqref{prob:mapping_optimization_problem} using simulated annealing (SA), a stochastic optimization method that is well suited to large-scale combinatorial problems with a complex search space.
Starting from a randomly generated initial configuration, SA iteratively explores neighboring configurations while gradually reducing the probability of accepting worse solutions, thereby converging to a near-optimal solution.

The neighborhood structure employed in SA consists of two types of moves, namely, reassigning a qubit or an MSF to an unoccupied patch, and swapping the assignments of two qubits or MSFs.
For a current configuration $\{x_i\}_{i=1}^n$ and a neighboring configuration $\{x'_i\}_{i=1}^n$, we evaluate the change in the objective function by considering only the components affected by the move.
Instead of recomputing the full objective value, we directly compute the incremental difference induced by the local modification.
This difference-based update enables each SA iteration to be computationally efficient.
With this implementation, the number of iterations is fixed to $1{,}000{,}000$, which can be executed within a few seconds for almost all tested instances.

When the movement of MSFs is allowed, corresponding to the inner case defined in \cref{subsec:sa_setting}, it becomes necessary to update the distance function $d_{\mathrm{MSF}}$ in Problem~\eqref{prob:mapping_optimization_problem}.
This update requires computing the distance from each data qubit to its nearest MSF.
Formally, this corresponds to a multiple-source shortest-paths problem on the chip graph, where all MSF locations act as sources.
This problem can be solved using a breadth-first search (BFS) algorithm, whose computational cost scales linearly with the number of patches.
As a result, even when MSF movements are included, the additional computations required by SA remain modest, and the overall optimization procedure retains its efficiency.

\subsection{Results for Mapping Evaluation}
\label{app:SA_results}

We present the numerical results of the mapping evaluation in \cref{fig:mapping_evaluation_results}, using the method described in \cref{app:mapping_how_to_solve}.
The evaluation was performed on four types of SELECT circuits: Fermi--Hubbard 2D, $\mathrm{Z}_2$ Lattice Gauge 2D, Schwinger, and Random Local.
The specific parameters for each Hamiltonian are written in \cref{tab:Hamiltonians}.

As discussed in \cref{sec:mapping}, the general trend shows that the SA approach yields the shortest total execution time in code beats. However, the objective function in Problem~\eqref{prob:mapping_optimization_problem}, optimized by SA, is not directly the total execution time $T_P$ (\cref{sec:performance_metric}).
Thus, in rare cases, a purely random assignment can even outperform the SA-based result.
Combining the scheduling method with direct checks of $T_P$ during SA is one of the promising directions for further improvement.
As suggested in recent work~\cite{molaviDependencyAwareCompilationSurface2025}, explicitly incorporating gate-operation dependencies into the optimization process could potentially enable more efficient scheduling and resource allocation.

\begin{figure}[t]
    \centering

    \begin{subfigure}
        [t]{0.5\columnwidth}
        \centering
        \includegraphics[width=\columnwidth]{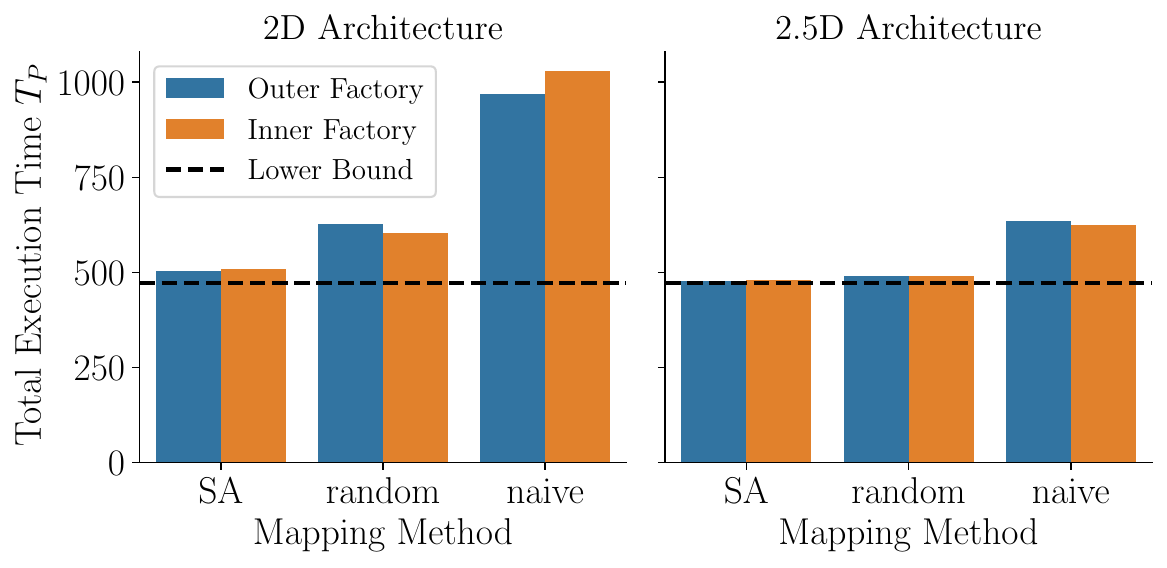}
        \subcaption{Fermi--Hubbard 2D}
    \end{subfigure}%
    \begin{subfigure}
        [t]{0.5\columnwidth}
        \centering
        \includegraphics[width=\columnwidth]{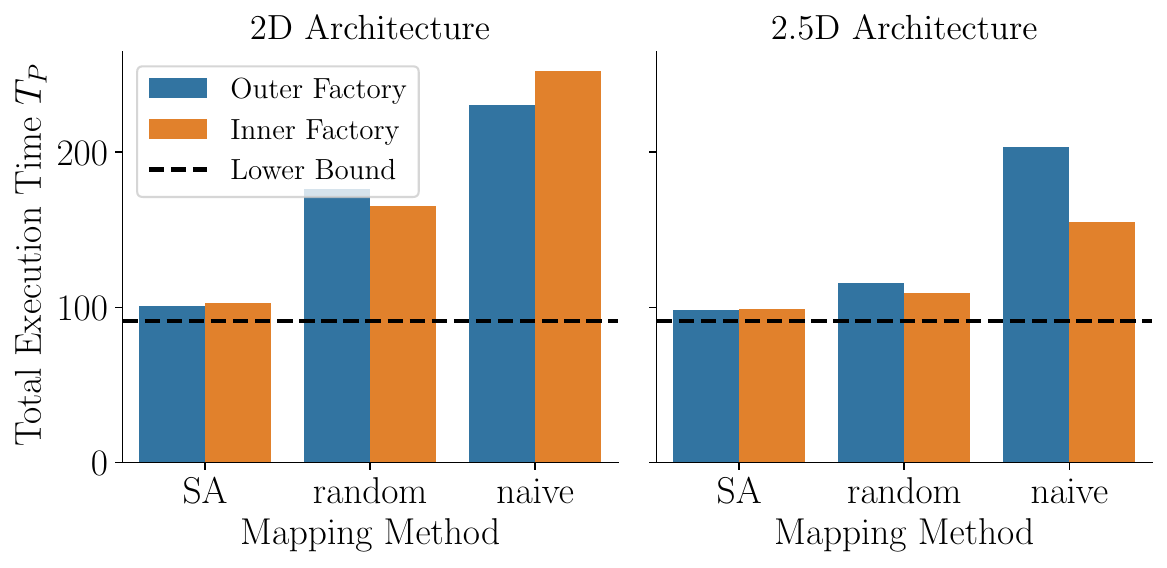}
        \subcaption{Z2 Lattice Gauge 2D}
    \end{subfigure}

    \begin{subfigure}
        [t]{0.5\columnwidth}
        \centering
        \includegraphics[width=\columnwidth]{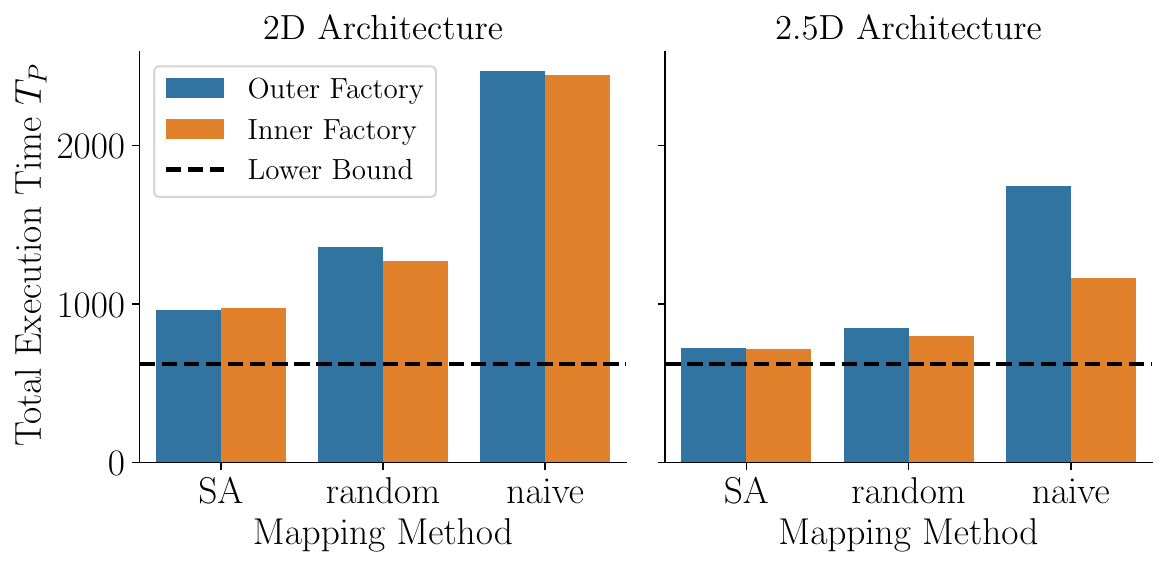}
        \subcaption{Schwinger}
    \end{subfigure}%
    \begin{subfigure}
        [t]{0.5\columnwidth}
        \centering
        \includegraphics[width=\columnwidth]{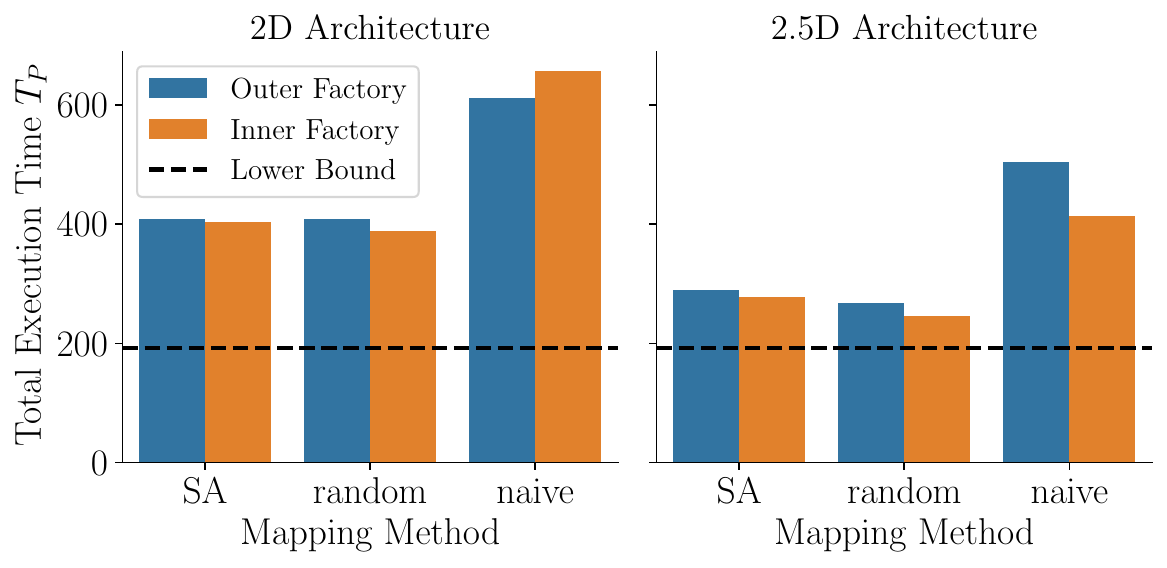}
        \subcaption{Random local}
    \end{subfigure}

    \caption{Mapping evaluation results for Hamiltonians for the part of SELECT circuits in \cref{tab:Hamiltonians}.
        The setting and baselines follow \cref{fig:total_time}.}
    \label{fig:mapping_evaluation_results}
\end{figure}

\section{Visualization Results}

We present some of the compilation visualizations in \cref{fig:individual}, following the visualization methodology for lattice surgery results described in \cref{fig:layout}.
The visualization code is publicly available on GitHub~\cite{github}.
Here, we focus on the results obtained using the projective and double routing algorithms under the outer 1D configuration, with SA employed for the mapping.

Several characteristic differences can be identified from these results.
The projective approach efficiently utilizes ancilla qubits and tends to reduce the total code beats by extending routing paths in advance.
This strategy often results in longer paths, which can lead to increased circuit volume and reduced spatial efficiency.
In contrast, the double and single approaches generally yield shorter routing paths.
Since each path occupies ancilla qubits for at most two time steps, fewer paths remain idle, which contributes to a smaller circuit volume.

\begin{figure*}
    \centering
    \includegraphics[width=0.78\textwidth]{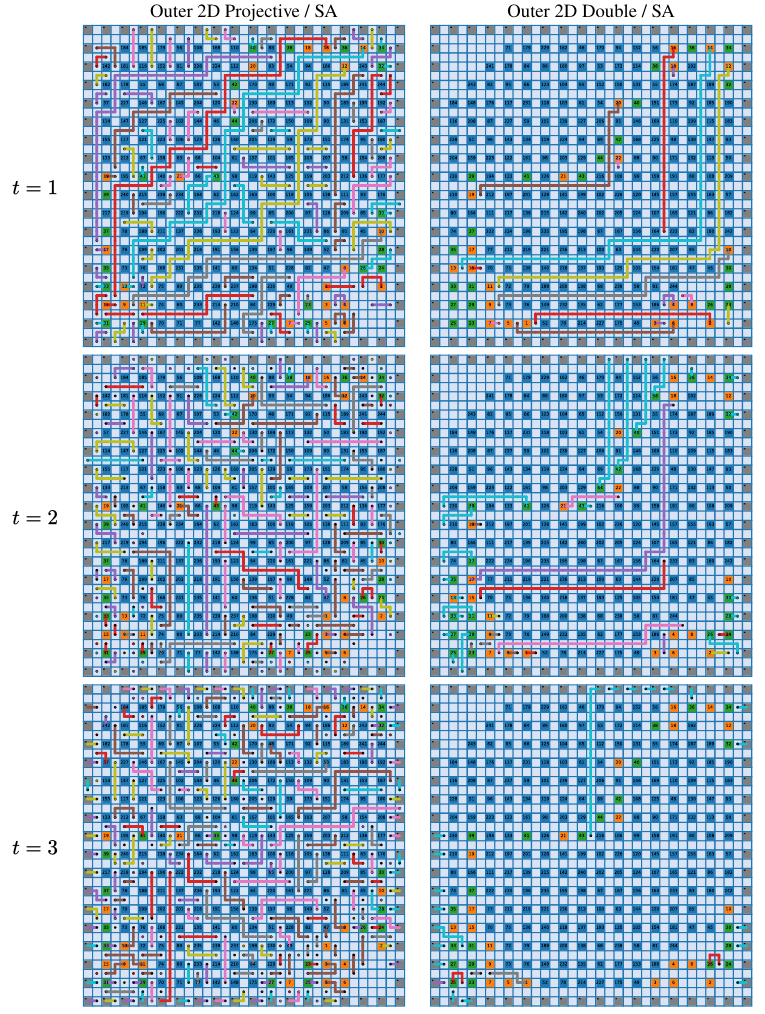}
    \caption{Visualization of compilation results obtained using the projective and double routing algorithms.
        We used the 2D outer configuration with SA mapping, and the SELECT circuit for the Fermi--Hubbard 2D Hamiltonian described in \cref{tab:Hamiltonians}.
        The visualization shows the first three code beats results ($t=1,2,3$).
        The projective approach effectively exploits ancilla qubits and typically requires fewer total code beats, at the cost of a larger circuit volume.
        In contrast, the double approach exhibits complementary behavior.}
    \label{fig:individual}
\end{figure*}

\end{document}